\PassOptionsToPackage{unicode}{hyperref}
\PassOptionsToPackage{hyphens}{url}
\PassOptionsToPackage{dvipsnames,svgnames,x11names}{xcolor}
\documentclass[
  12pt]{article}

\usepackage{amsmath,amssymb}
\usepackage{amsfonts}
\usepackage{mathtools}
\usepackage{amsthm}

\newtheorem{theorem}{Theorem}
\newtheorem{lemma}{Lemma}

\newtheorem{proposition}{Proposition}
\newtheorem{assumption}{Assumption}
\newtheorem{definition}{Definition}
\usepackage[toc,page]{appendix}
\usepackage{iftex}
\ifPDFTeX
  \usepackage[T1]{fontenc}
  \usepackage[utf8]{inputenc}
  \usepackage{textcomp} % provide euro and other symbols
\else % if luatex or xetex
  \usepackage{unicode-math}
  \defaultfontfeatures{Scale=MatchLowercase}
  \defaultfontfeatures[\rmfamily]{Ligatures=TeX,Scale=1}
\fi
\usepackage{lmodern}
\ifPDFTeX\else  
\fi
\IfFileExists{upquote.sty}{\usepackage{upquote}}{}
\IfFileExists{microtype.sty}{% use microtype if available
  \usepackage[]{microtype}
  \UseMicrotypeSet[protrusion]{basicmath} % disable protrusion for tt fonts
}{}
\makeatletter
\@ifundefined{KOMAClassName}{% if non-KOMA class
  \IfFileExists{parskip.sty}{%
    \usepackage{parskip}
  }{% else
    \setlength{\parindent}{0pt}
    \setlength{\parskip}{6pt plus 2pt minus 1pt}}
}{% if KOMA class
  \KOMAoptions{parskip=half}}
\makeatother
\usepackage{xcolor}
\makeatletter
\ifx\paragraph\undefined\else
  \let\oldparagraph\paragraph
  \renewcommand{\paragraph}{
    \@ifstar
      \xxxParagraphStar
      \xxxParagraphNoStar
  }
  \newcommand{\xxxParagraphStar}[1]{\oldparagraph*{#1}\mbox{}}
  \newcommand{\xxxParagraphNoStar}[1]{\oldparagraph{#1}\mbox{}}
\fi
\ifx\subparagraph\undefined\else
  \let\oldsubparagraph\subparagraph
  \renewcommand{\subparagraph}{
    \@ifstar
      \xxxSubParagraphStar
      \xxxSubParagraphNoStar
  }
  \newcommand{\xxxSubParagraphStar}[1]{\oldsubparagraph*{#1}\mbox{}}
  \newcommand{\xxxSubParagraphNoStar}[1]{\oldsubparagraph{#1}\mbox{}}
\fi
\makeatother

\usepackage{longtable,booktabs,array}
\usepackage{calc} % for calculating minipage widths
\usepackage{etoolbox}
\makeatletter
\patchcmd\longtable{\par}{\if@noskipsec\mbox{}\fi\par}{}{}
\makeatother
\IfFileExists{footnotehyper.sty}{\usepackage{footnotehyper}}{\usepackage{footnote}}
\makesavenoteenv{longtable}
\usepackage{graphicx}
\makeatletter
\def\maxwidth{\ifdim\Gin@nat@width>\linewidth\linewidth\else\Gin@nat@width\fi}
\def\maxheight{\ifdim\Gin@nat@height>\textheight\textheight\else\Gin@nat@height\fi}
\makeatother
\setkeys{Gin}{width=\maxwidth,height=\maxheight,keepaspectratio}
\makeatletter
\def\fps@figure{htbp}
\makeatother

\makeatletter
\@ifpackageloaded{caption}{}{\usepackage{caption}}
\AtBeginDocument{%
\ifdefined\contentsname
  \renewcommand*\contentsname{Table of contents}
\else
  \newcommand\contentsname{Table of contents}
\fi
\ifdefined\listfigurename
  \renewcommand*\listfigurename{List of Figures}
\else
  \newcommand\listfigurename{List of Figures}
\fi
\ifdefined\listtablename
  \renewcommand*\listtablename{List of Tables}
\else
  \newcommand\listtablename{List of Tables}
\fi
\ifdefined\figurename
  \renewcommand*\figurename{Figure}
\else
  \newcommand\figurename{Figure}
\fi
\ifdefined\tablename
  \renewcommand*\tablename{Table}
\else
  \newcommand\tablename{Table}
\fi
}
\@ifpackageloaded{float}{}{\usepackage{float}}
\floatstyle{ruled}
\@ifundefined{c@chapter}{\newfloat{codelisting}{h}{lop}}{\newfloat{codelisting}{h}{lop}[chapter]}
\floatname{codelisting}{Listing}

\makeatother
\makeatletter
\@ifpackageloaded{caption}{}{\usepackage{caption}}
\@ifpackageloaded{subcaption}{}{\usepackage{subcaption}}
\makeatother

\ifLuaTeX
  \usepackage{selnolig}  % disable illegal ligatures
\fi
\usepackage[]{natbib}
\usepackage{bookmark}

\IfFileExists{xurl.sty}{\usepackage{xurl}}{} % add URL line breaks if available
\hypersetup{
  pdftitle={Title},
  pdfauthor={Author 1; Author 2},
  pdfkeywords={3 to 6 keywords, that do not appear in the title},
  colorlinks=true,
  linkcolor={black},
  filecolor={Maroon},
  citecolor={black},
  urlcolor={black},
  pdfcreator={LaTeX via pandoc}}

\newcommand{\anon}{1}

\newcommand{\E}{\mathrm{E}}
\newcommand{\N}{\mathbb{N}}
\newcommand{\Z}{\mathbb{Z}}
\newcommand{\bh}{\boldsymbol{h}}
\newcommand{\bl}{\boldsymbol{\lambda}}

\newcommand{\Var}{\mathrm{Var}}
\newcommand{\cum}{\textnormal{cum}}

\newcommand{\argmax}[1]{\underset{#1}{\operatorname{argmax}}}

\usepackage{textgreek}

\begin{document}

\def\spacingset#1{\renewcommand{\baselinestretch}%
{#1}\small\normalsize} \spacingset{1}

%%%%%%%%%%%%%%%%%%%%%%%%%%%%%%%%%%%%%%%%%%%%%%%%%%%%%%%%%%%%%%%%%%%%%%%%%%%%%%

\if1\anon
{
  \title{\bf Online Spectral Density Estimation}
  \author{Shahriar Hasnat Kazi\thanks{E-mail: shahriar.kazi16@imperial.ac.uk}\\
  \textit{Department of Mathematics, Imperial College London,}\\
  Niall Adams\\
  \textit{Department of Mathematics, Imperial College London} and\\ Edward A. K. Cohen\\
  \textit{Department of Mathematics, Imperial College London}}
  \maketitle
} \fi

\if0\anon
{
  \bigskip
  \bigskip
  \bigskip
  \begin{center}
    {\LARGE\bf Title}
\end{center}
  \medskip
} \fi

\bigskip
\begin{abstract}
This paper develops the first online algorithms for estimating the spectral density function --- a fundamental object of interest in time series analysis --- that satisfies the three core requirements of streaming inference: fixed memory, fixed computational complexity, and temporal adaptivity. Our method builds on the concept of forgetting factors, allowing the estimator to adapt to gradual or abrupt changes in the data-generating process without prior knowledge of its dynamics. We introduce a novel online forgetting-factor periodogram and show that, under stationarity, it asymptotically recovers the properties of its offline counterpart. Leveraging this, we construct an online Whittle estimator, and further develop an adaptive online spectral estimator that dynamically tunes its forgetting factor using the Whittle likelihood as a loss. Through extensive simulation studies and an application to ocean drifter velocity data, we demonstrate the method's ability to track time-varying spectral properties in real-time with strong empirical performance.
\end{abstract}

\noindent%
{\it Keywords:} Spectral density function, online estimation, forgetting factors, time series analysis
\vfill

\newpage
\spacingset{1.8} % DON'T change the spacing!

\section{Introduction}

\label{sect: Intro}

With data being collected at unprecedented volumes and velocities, the analysis of dynamic, continuous data streams presents challenges that traditional \emph{offline} algorithms cannot address. These algorithms rely on storing and preprocessing large datasets, resulting in prohibitive memory and time requirements that make them unsuitable for scenarios requiring real-time, up-to-date results. In contrast, \emph{online algorithms} process data incrementally, as it arrives, operating under strict constraints of fixed memory and computational complexity while avoiding the need to revisit past data \citep{cit4, cit20, cit32}. Beyond the challenges of scale and speed, data streams are also inherently dynamic, as the processes generating the data often evolve over time. This evolving nature introduces a fundamental challenge that online algorithms must also address. For example, if the algorithm is estimating a target quantity of interest from the data, that target may itself be changing over time.

Concept drift refers to changes in time of the data-generating process, which can occur gradually or abruptly. If not addressed, such changes can significantly degrade the performance of estimation algorithms that rely on the assumption of strict or weak stationarity (in the probabilistic sense). To handle this, \emph{temporal adaptivity} becomes essential— the ability of an algorithm to adjust to concept drift while maintaining performance. Importantly, temporal adaptivity must not compromise accuracy and precision when the data-generating process remains stable; under such conditions, the algorithm’s performance should approach that of a stationary estimator. Along with fixed memory and computational efficiency, temporal adaptivity is a core requirement for online algorithms, enabling real-time analysis of dynamic systems.

In this paper, our target of interest is the spectral density function (SDF), a fundamental concept in time series analysis (see \cite{PSD_ref_1} for a comprehensive theoretical and methodological overview). Defined as the Fourier transform of the autocovariance sequence of a weakly stationary process, the SDF provides an information-rich representation of a process’s second-order structure in the frequency domain, revealing periodic behaviour and the dominant frequencies driving its variability. This frequency-domain representation is often preferred over its time-domain counterpart, particularly when a process's variability is linked to underlying physical phenomena or periodic components. As a result, the SDF has been widely used in fields such as astronomy \citep{astr_ref}, cognitive science \citep{cogn_sc_ref}, earth sciences \citep{PSD_ref_2, earth_ref}, electrical engineering \citep{EE_ref}, and finance \citep{finance_ref}. Its estimation remains a topic of significant interest \citep{recent_ref, recent_ref_2, recent_ref_3, DebiasedWhittle}, and its use is expanding to new domains \citep{cohen21,grainger24}.

While the SDF is defined with respect to stationary processes, in reality, data rarely conforms to this assumption; the spectral properties of the data generating process may gradually drift or undergo abrupt changes. This reality motivated the now mature field of time-frequency analysis, both in the formulation of time series models that emit a time-varying spectrum and in computational methods, including the short-time Fourier transform \citep{short_FT_1, short_FT_2, short_FT_3} and the widely celebrated discovery of wavelets \citep{wavelets, wavelets_2}. However, these and subsequent approaches remain offline and are unsuitable for modern applications in which real-time tracking is needed as new data points arrive. We address this shortfall by developing the first online algorithm for the estimation of the SDF that adheres to the requirements of fixed memory, fixed complexity and temporal adaptivity.

To do so, we take the \emph{forgetting factors} approach, popularised in areas such as change-point detection \citep{FF_examples_CPdet_1, FF_examples_CPdet_2}, classification \citep{FF_examples_class_1, FF_examples_class_3}, clustering \citep{FF_examples_clust_1, FF_examples_clust_2}, density estimation \citep{FF_examples_distr}, pattern mining \citep{FF_examples_mining_1, FF_examples_mining_2} and model validation \citep{FF_examples_model}. Forgetting factors are a sequence of scalars taking values in $(0, 1]$, controlling the rate at which past data is forgotten. Forgetting factors provide a continuous analogue to hard forgetting approaches, where each data point contributes equally, or not at all, to the estimator \citep{FF_Christoforos}. The most immediate alternative to forgetting factors is a sliding window. However, specifying the size of the window is non-trivial. In fact, a commonality between the approaches to ensure temporal adaptivity introduced in the literature is the requirement of specifying control parameters. The selection of suitable values for these parameters is a non-trivial task and frequently requires prior knowledge of the system, making these approaches unsuitable when no such information is available. This selection task is at the core of the \textit{adaptive} streaming estimation literature. Of particular relevance to this paper are adaptive forgetting factors \citep{cit7, FF_examples_CPdet_2}, which we incorporate into our algorithms so that the online spectral estimator can learn a local optima for the forgetting factor based on the present behaviour of the data stream it is analysing.

The periodogram remains the most widely used approach to estimating the SDF and sits at the centre of most other estimation schemes that set about to improve upon its bias and variance properties. This includes direct spectral estimators, in which a data taper is used to control bias. In Section 2, after reviewing these in the traditional offline setting, we present a novel online forgetting-factor periodogram. Although the idea of a sequential algorithm for the estimation of the spectral density of a process has found some development in the literature \citep{recursiveslidingwindow}, our method differs in its approach, since we tackle the problem from a tapering angle. Importantly, we show that under the assumption of stationarity it asymptotically approaches the same distribution as its offline counterpart, while retaining fixed and finite complexity, with each data point only used once, upon arrival. To simplify the exposition, we present the forgetting-factor periodogram for a process with a known and constant mean; without loss of generality, this can be assumed to be zero. We then show how the forgetting factor periodogram can be readily extended to processes with an unknown and potentially time varying mean.

In the offline setting, the Whittle estimator has been studied and applied extensively \citep{Whittle_ref, Whittle_ref_2, Whittle_ref_3} for parameteric spectral estimation. The asymptotic distribution of the periodogram is used as a finite data approximation to construct a likelihood for the periodogram in terms of parameters of a specified model (e.g. autoregressive process). This can then be optimised via standard techniques. In Section 3, we use the forgetting-factor periodogram and its derived statistical properties, to develop an online Whittle estimator. In Section 4, we then leverage the Whittle likelihood as a loss function for adaptive forgetting to give an adaptive online spectral estimator.

We evaluate the performance of our online estimators using both simulated and real data. In Section 5, we present a series of simulation experiments that highlight the temporal adaptivity of the algorithms. In Section 6, we apply the method to a time series of drifter velocity data from the Global Drifter Programme, showcasing its capability to track the drifter's position in real-time.

%**** Need to think of best place to put this. 
%Canonically, when working with the Periodogram and any of its variations, the assumption of a known or zero mean is ubiquitous. Of course, this assumption is rarely representative of real data. Nevertheless, as we focus on estimating second-order structures in the data generating process, we will make this assumption in the following derivations. For the empirical study in Section \ref{sect: ocean data}, we center the data before applying our framework.

\section{Forgetting Factor Periodogram}
\label{sect: FFP}
\subsection{Processes with a known constant mean}
Let $\{X_t\}_{t\in\Z}$ be a weakly-stationary time series, meaning that
\begin{align*}
	\E[X_t] &= \mu \quad \quad \quad \forall t \in\Z,\\
	\text{Cov}[X_t, X_{t+\tau}] &= \text{Cov}[X_0 , X_\tau] \quad \forall \tau\in\N \quad\forall t \in\Z.
\end{align*}
Let $s_\tau = \text{Cov}[X_t, X_{t+\tau}]$, for $\tau \in \Z$, be the autocovariance sequence. Assuming $\sum_{\tau = -\infty}^\infty |s_\tau| < \infty$, the SDF $S(f)$ of a zero-mean weakly stationary process exists, is continuous and bounded, and is defined as
\begin{equation}
	S(f) = \sum_{\tau = -\infty}^\infty s_\tau e^{-i 2\pi f \tau}, \quad  f \in (-1/2, 1/2].
	\label{eq: spectral density definition}
\end{equation}
The SDF describes the contribution oscillations at frequency $f$ have to the overall variance of a random process. While it encodes equivalent information to the autocovariance sequence (the Fourier transform is bijective), it often offers a richer environment in which to study the second-order-structure. 

Spectral density estimation methods can be broadly classified as parametric or nonparametric. Parametric approaches are primarily based on autoregressive moving average (ARMA) models, but they tend to give misleading inferences when the parametric model is poorly specified \citep{miss1, miss2}. Nonparametric models often rely on the periodogram. For a portion of a zero-mean process $X_1$,...,$X_T$, the periodogram is defined as 
\begin{equation*}
	\widehat{S}_T(f) = \frac{1}{T} \left | \sum_{t = 1}^T X_t e^{-i 2\pi f t} \right |^2, \quad  f \in (-1/2, 1/2].
\end{equation*}
This can be efficiently computed with the fast Fourier transform (FFT).  Note that should $\mu$ be non-zero but known, the time series must be \emph{centered} (have the mean substracted from it) before the periodogram is computed. Failure to do so will result in a significant spike in the periodogram at $f=0$ which can obscure important features.

The periodogram is known to be a biased estimator of the SDF \citep{pukkila1979bias}. This occurs due to a blurring caused by having only a finite sample of the infinite sequence. More specifically,
$$
	\E[\widehat{S}_T(f)] = \int_{-1/2}^{1/2} \mathcal{F}_T(f-f')S(f')d f',
$$
where 
$$
	\mathcal{F}_T(f) = \frac{\sin^2(T\pi f)}{T\sin^2(\pi f)},
$$
is Fejer's Kernel. This effect is sometimes also referred to as a `sidelobe leakage' caused by the poor frequency concentration properties of $\mathcal{F}_T(f)$; for details see \cite{Fejer}. Tapering is a technique used to reduce the kernel's sidelobes. Given a sequence of real valued constants $\bh_T \equiv (h_t)_{t = 1, ..., T}$ such that
$$
	\sum_{t = 1}^T h_t^2 = 1,
$$
the tapered spectral density estimator, often referred to as the direct spectral estimator, is defined as
\begin{equation*}
	\widehat{S}_{T}(f;\bh_T) = \left | \sum_{t = 1}^T h_t X_t e^{-i 2\pi f t} \right |^2, \quad  f \in (-1/2, 1/2].
\end{equation*}
Analogously, we have 
$$
	\E[\widehat{S}_{T}(f;\bh_T)] = \int_{-1/2}^{1/2} \mathcal{H}_T(f-f')S(f')d f',
$$
where $\mathcal{H}_T(f) = |\sum_{t = 1}^T h_t e^{-i 2\pi f t}|^2$, motivating a choice of taper that has a $\mathcal{H}_T(f)$ kernel with better frequency concentration than $\mathcal{F}_T(f)$. An alternative view of a taper is as a tool to control the influence of each datum to the computation of the spectral density estimate. This is analogous to the sequence of scalars used in the forgetting factor paradigm. For a forgetting factor $\lambda$, the taper is defined
$$
h_{t} = C_T(\lambda)^{-1/2} \lambda ^ {T-t},\quad t=1,...,T
$$
where constant $C_T(\lambda)$ is chosen such that $\sum_{t=1}^T h_t^2 = 1$, giving
$$
	h_t = \lambda^{T-t} \sqrt{\frac{1-\lambda^2}{1-\lambda^{2T}}},\quad t=1,...,T.
$$
From this, we propose the \textit{forgetting factor periodogram} (FFP)
$$
	\widehat{S}_T(f;\lambda) = \frac{1}{C_T(\lambda)}\left | \sum_{t = 1}^T \lambda^{T-t} X_t e^{-i 2\pi f t} \right |^2,
$$
where $f \in (-1/2, 1/2]$ and $$C_T(\lambda) = \frac{1-\lambda^{2T}}{1-\lambda^2}.$$ 

To show this fulfils the requirements of an online algorithm, define
$$
	J_T (f;\lambda) = \sum_{t = 1}^T \lambda^{T-t} X_t e^{-i 2\pi f t},
$$
such that $\widehat{S}_T(f;\lambda) = |J_T(f;\lambda)|^2/C_T(\lambda)$. For a new data point $X_{T+1}$, 
\begin{align}
	J_{T+1}(f;\lambda) &=\sum_{t = 1}^{T+1} \lambda^{T+1-t} X_t e^{-i 2\pi f t}\nonumber\\
	&= \lambda\sum_{t = 1}^{T} \lambda^{T-t} X_t e^{-i 2\pi f t} + X_{T+1} e^{-i 2\pi f (T+1)}\nonumber\\
	&= \lambda J_T(f; \lambda) + X_{T+1} e^{-i 2\pi f (T+1)}
	\label{eq: J update}
\end{align}
and, similarly,
\begin{equation}
	C_{T+1}(\lambda) =\sum_{t = 1}^{T+1}\lambda^{2(T+1-t)}
	= \lambda^2\sum_{t = 1}^{T} \lambda^{2(T-t)}  + 1
	= \lambda^2 C_T(\lambda) + 1.
	\label{eq: C update}
\end{equation}
Therefore, to update $\widehat{S}_T(f;\lambda)$ to $\widehat{S}_{T+1}(f;\lambda)$, we only need to store the values of $J_T(f;\lambda)$ and $C_T(\lambda)$; each datum is used only once to perform computations (\ref{eq: J update}) and (\ref{eq: C update}), satisfying $O(1)$ time complexity. To satisfy $O(1)$ memory complexity, we only store the values of $\widehat{S}_T(f;\lambda)$ at a discrete and fixed $M<\infty$ number of frequencies, as is common in offline spectral estimation, which typically uses the FFT algorithm. This means, at any one time, storing $M+1$ values: $J_T (f;\lambda)$ for the $M$ frequencies of interest and $C_T(\lambda)$.

\subsection{Processes with unknown mean}
If the mean of the process is unknown, and may even be varying with time, then it needs to be estimated and subtracted from the process before the SDF is estimated. Here, we show that this can be implemented within the forgetting factor regime and does not prevent the derivation of a recursive formula for the DFT and periodogram. 

Let us expand the notation for the forgetting factor periodogram,
\begin{align*}
    J_{X,T} (f;\lambda) &= \sum_{t = 1}^T \lambda^{T-t} X_t e^{-i 2\pi f t},\\
    \widehat{S}_{X,T}(f;\lambda) &= \frac{|J_{X,T} (f;\lambda)|^2}{C_T(\lambda)},
\end{align*}
where we now include a subscript to indicate the process that is being operated on. Define $\mu_{X,t} = \mathbb{E}[X_t]$ and $\bar{X}_t$ as an estimate for $\mu_{X,t}$. We have demonstrated how $\widehat{S}_{X,T}(f;\lambda)$ can be computed sequentially. The concern is whether the same is true for $\widehat{S}_{Y,T}(f;\lambda)$, where $\{Y_t\}$ is a sequentially centred process $Y_t = X_t-\bar{X}_t$. Note that
\begin{align*}
    J_{Y,T} (f;\lambda) &= \sum_{t = 1}^T \lambda^{T-t} [X_t-\bar{X}_t] e^{-i 2\pi f t},\\
    &= \sum_{t = 1}^T \lambda^{T-t} X_t e^{-i 2\pi f t}-\sum_{t = 1}^T \lambda^{T-t} \bar{X}_t e^{-i 2\pi f t},\\
    &= J_{X,T} (f;\lambda) - J_{\bar X,T} (f;\lambda).
\end{align*}
As shown previously, $J_{X,T} (f;\lambda)$ can be updated sequentially and so can $J_{\bar X,T} (f;\lambda)$ as long as $\bar{X}_T$ allows for a recursive formula. This is indeed the case for the sample mean,
\begin{align*}
    \bar{X}_{T+1} &= \frac{\sum_{t=1}^{T+1} X_t}{T+1},\\
    &= \frac{T}{T+1}\bar{X}_{T} + \frac{X_{T+1}}{T+1}.
\end{align*}
Moreover, the sample mean can be handled with a similar forgetting factor to the periodogram, resulting in a weighted mean,
\begin{align*}
    \bar{X}_{T+1}(\lambda) &= \frac{\sum_{t=1}^{T+1} \lambda^{T+1-t}X_t}{D_{T+1}(\lambda)},\\
    &= \frac{\lambda D_{T}(\lambda)}{D_{T+1}(\lambda)}\bar{X}_{T} + \frac{X_{T+1}}{D_{T+1}(\lambda)},
\end{align*}
where $D_{T+1}(\lambda) = \sum_{t=1}^{T+1}\lambda^{T+1-t} = \lambda D_{T}(\lambda) +1$. While this forgetting factor does not need to be the same as the one used in the computation of the forgetting factor periodogram, the role of the two hyper-parameters as indicators of how fast the process is changing would suggest using the same value as a suitable simplification.

\section{Online Whittle Estimation}
\subsection{Whittle Estimation}

The Whittle likelihood for a stationary (not necessarily Gaussian) process is motivated by an asymptotic approximation \citep{hannan1973, whittle1961} to the distribution of the periodogram, or related non-parametric estimator. In the offline setting, for a centered process, consider
\begin{equation*}
	J_T(f) = \sum_{t = 1}^T X_t e^{-i 2\pi f t}, \quad  f \in (-1/2, 1/2],
\end{equation*}
such that the periodogram $\hat S_T (f) = |J_T(f)|^2/T$. Let $\text{Cum}(X_{t+u_1},X_{t+u_2},...,X_{t+u_{k-1}},X_t)$ denote the $k$th order cumulant function and assume for all $k\geq 2$, $u_1,...,u_{k-1}$ and $t$ that $\text{Cum}(X_{t+u_1},X_{t+u_2},...,X_{t+u_{k-1}},X_t) = \text{Cum}(X_{u_1},X_{u_2},...,X_{u_{k-1}},X_0)$ (i.e. strict-stationarity), in addition to the standard mixing condition \citep{brillinger2001},
\begin{equation}
	\label{mixing}
		\sum_{u_1 = -\infty}^\infty\sum_{u_2 = -\infty}^\infty...\sum_{u_{k-1} = -\infty}^\infty |c_k(u_1, u_2, ..., u_{k-1})| < \infty,
\end{equation}
for all $k\geq 2$, where $c_k(u_1, u_2, ..., u_{k-1})\equiv \text{Cum}(X_{u_1},X_{u_2},...,X_{u_{k-1}},X_0)$. Then, for $f\notin \{-1/2, 0, 1/2\}$, $${T^{-1/2} J_T(f) \longrightarrow \mathcal{CN}(0, S(f))},$$ as $T\longrightarrow\infty$, where $\mathcal{CN}(\mu, \Sigma)$ denotes a circularly symmetric complex normal distribution. Furthermore, for Fourier frequencies $f_k = k/T$, then $\{T^{-1/2} J(f_k);k=1,...,\lfloor T/2\rfloor\}$ are asymptotically pairwise independent. Hence, $\{\widehat{S}_T(f_k) = T^{-1}|J(f_k)|^2 ;  k=1...\lfloor T/2\rfloor\}$ behave as asymptotically independent exponential random variables with means $S(f_k)$, and equivalently, $\{2 \widehat{S}_T(f_k)/S(f_k); k=1...\lfloor T/2\rfloor\}$ are asymptotically independent $\chi^2_2$ random variables \citep{priestley1981}.% In the Gaussian white noise case, $\{\widehat{S}_T(f_k); k=1,...,\lfloor T/2\rfloor\}$ are independent . 
 Guided by this, an approximate log-likelihood for $\{\widehat{S}_T(f_k);k=1,...,\lfloor  T/2\rfloor\}$ is 
\begin{equation}
	\mathcal{L}_T(S) = -\frac{1}{T} \sum_k \left (\log S(f_k) + \frac{\widehat{S}_T(f_k)}{S(f_k)} \right).
	\label{eq: Whittle Likelihood}
\end{equation}
This is the Whittle likelihood approximation. Note that for fixed $S$, as $T\longrightarrow\infty$, 
$$
	\E\left[\mathcal{L}_T(S)\right] \longrightarrow \E\left[\mathcal{L}(S)\right] = -\int_{-1/2}^{1/2} \left (\log S(f) + \frac{S^*(f)}{S(f)} \right) d f,
$$
where $S^*$ is the true spectrum. Therefore, because $\E[\mathcal{L}(S)]$ is maximized at $S = S^*$, $\mathcal{L}_T(S)$ should be a reasonable objective function for the identification of the spectrum. This is a key property of the Whittle likelihood from which consistency results can be developed \citep{hannan1973}. From here on, \textit{Whittle estimation} will refer to the inference procedure where, for a parametric family of SDFs $\{S_\phi(f);\phi\in\Phi\}$ parameterised by $\phi$ in parameter space $\Phi$, the \textit{Whittle estimator} is
$$
	\widehat{\phi}_{T} = \argmax{\phi\in\Phi}\  \mathcal{L}_{T}(S_\phi).
$$
In later sections, we consider the ARMA processes, or more specifically the autoregressive process $\{X_t\}$ of the form
$X_t = \sum_{i=1}^p\phi_i X_{t-i} + \epsilon_t$ where $\epsilon_t\sim \mathcal{N}(0, \sigma^2)$. Here, the target of inference is $\phi = (\phi_1, ..., \phi_p, \sigma^2)$. The Whittle likelihood is useful because it involves $S$ directly, in contrast to the exact likelihood, which involves $S$ indirectly through the autocovariances. Results for autoregressive moving average models on the higher-order efficiency of maximising (\ref{eq: Whittle Likelihood}) are provided in \cite{taniguchi1987} and \cite{taniguchi2002}. Similar results for long-memory time series were established by \cite{fox1986large}, \cite{dahlhaus1989efficient} and \cite{giraitis1990central} in parametric models and by \cite{robinson1994semiparametric} in semiparametric models. In a nonparametric set-up, Whittle likelihood has also been used for automated spline smoothing of the periodogram \citep{wahba1980automatic}, sieve maximum likelihood estimation \citep{chow1985sieve}, penalized maximum likelihood estimation \citep{Whittle_ref_2}, polynomial spline fitting \citep{kooperberg1995rate}, local-likelihood-based estimation \citep{fan1998automatic} and Bayesian estimation of the spectral density function \citep{carter1997semiparametric}.

\subsection{Forgetting factor Whittle estimation}
\label{sect: Whittle Estimation with FFP}
Having defined the Whittle likelihood in the offline setting as an approximate likelihood for the periodogram, we seek to use the same approximation for the FFP in the online setting. However, we need to show that the limiting distribution of $\widehat{S}_T(f;\lambda)$ as $T \rightarrow \infty$, is appropriate. In Appendix A, we formally set out and prove a theorem which shows that $$\left[C_T(\bl_{T-1})\right]^{-1/2} J_T(f;\bl_{T-1})\sim  \mathcal{CN}(0, S(f))$$  as $T\rightarrow\infty$, under the same assumptions of strict stationarity and meeting the mixing condition  (\ref{mixing}). Note there is a requirement in the asymptotic framework that the forgetting factor has a time dependency and simultaneously must tend to $1$ as $T \rightarrow \infty$; this induces the dependency on a sequence $\bl_{T-1}\equiv(\lambda_t)_{t=1,...,T-1}$ of time-varying forgetting factors, defined up until time $(T-1)$. This is needed to have an asymptotically growing sample size, but shows that the Whittle likelihood is a valid approximation in the forgetting regime and further motivates the development of an adaptive forgetting factor periodogram, in which the forgetting factor does increase through time when the process remains stationarity. This is presented in Section \ref{sect: Adap FFP} where the notation and time-varying forgetting factor estimators for the SDF are further explained and developed.

Remaining in the fixed forgetting factor setting for the moment, define the (fixed) Forgetting Factor Whittle Likelihood as
\begin{equation}
	\mathcal{L}_T(S;\lambda) = -\frac{1}{T} \sum_k \left (\log S(f_k) + \frac{\widehat{S}_T(f_k;\lambda)}{S(f_k)} \right),
	\label{eq: FF Whittle Likelihood}
\end{equation}
such that $\mathcal{L}_T(S;1) \equiv \mathcal{L}_T(S)$. Given a parametric family of SDFs, $\{S_\phi|\phi\in \Phi\}$, define the \textit{Forgetting Factor Whittle Estimator} (FFWE), $\widehat{\phi}_T^{(\lambda)}$, as
\begin{equation}
	\widehat{\phi}_T(\lambda) = \argmax{\phi}\ \mathcal{L}_T(S_\phi;\lambda).
	\label{eq: argmax L}
\end{equation}
The advantage of using $\widehat{\phi}_T(\lambda)$ as the estimator for the parameters of a process is that it inherits the temporal adaptivity of the FFP; the re-weighting of the data occurs when computing $\widehat{S}_T(f_k;\lambda)$. As this quantity is the only way for the data to enter the computation of $\mathcal{L}_T(S_\phi;\lambda)$, the re-weighting automatically adjusts the likelihood and thus the parameter estimate $\widehat{\phi}_T(\lambda)$.

Solving (\ref{eq: argmax L}) often involves a non-linear maximization problem. Finding a solution in this way for each new observation violates the complexity constraints for the algorithm to be truly online. We therefore opt to instead use gradient descent to approximate a solution for the $(T+1)^{th}$ datum based on the solution for the $T^{th}$ datum.

Assume we have a sequence $\{X_t\}$ and parameter estimates $\{\hat{\phi}_t(\lambda); t = 1, ..., T\}$. For a new observation $X_{T+1}$ we approximate $\hat{\phi}_{T+1}(\lambda)$ with
\begin{equation*}
	\widehat{\phi}_{T+1}(\lambda) = \widehat{\phi}_T(\lambda) + r_\phi\nabla_\phi \mathcal{L}_T\left(S_{\widehat{\phi}_T(\lambda)};\lambda\right),
\end{equation*}
where $r_\phi$ is the learning rate. Typically, the learning rate for a gradient descent algorithm is used to control the speed at which we want our estimated solution to approach the true solution \citep{gradient_descent}. However, in the case of data streams, where the data generating process may undergo some conceptual drift, $r_\phi$ further serves as a control parameter for how quickly the FFWE can adapt to changes in the process.

\section{Adaptive Forgetting Factor Spectral Estimation}
\label{sect: Adap FFP}

This forgetting factor scheme, like many others, suffers from a major drawback: there is no clear way to set $\lambda$ a priori. This motivates adaptive forgetting factors where we let an online, adaptive method tune the forgetting factor to the data. The method we use is based on stochastic gradient descent and our construction requires that at time $T$ we already have a time-dependent sequence of forgetting factors $\bl_{T-1} = (\lambda_t)_{t=1,...,T-1}$. In practice, a burn-in period is used to initialise this.

Given $\bl_{T-1}$ have been pre-defined, then our (unnormalised) data-taper becomes
$$
h_t = \left\{\begin{array}{ll}
	1 & t=T; \\
	\prod_{s=t}^{T-1}\lambda_s & t = 1,...,T-1; \\
	0 & \text{otherwise.}
\end{array}
\right. 
$$
From this, the FFP $\hat{S}_T(f;\bl_{T-1})$ becomes 
$$
\hat{S}_T(f;\bl_{T-1}) = \frac{1}{C_T(\bl_{T-1})}\left|J_T(f;\bl_{T-1}) \right|^2
$$
where $J_T(f;\bl_{T-1}) = \sum_{t=1}^T h_t X_t e^{-i 2\pi ft}$ and $C_T(\bl_{T-1})= \sum_{t=1}^T h^2_t$. Using these, the definition of the adaptive forgetting factor Whittle likelihood $\mathcal{L}_T(S;\bl_{T-1})$ follows naturally from (\ref{eq: FF Whittle Likelihood}). Furthermore, the online update rules for $J_T(f;\bl_{T-1})$ and $C_T(\bl_{T-1})$  are derived in an analogous manner to their fixed forgetting factor counterparts, giving 
\begin{align}
J_{T+1}(f;\bl_{T}) & = \lambda_T J_T(f;\bl_{T-1}) + X_{T+1}e^{i2\pi f(T+1)} \label{updateJAFF}\\
C_{T+1}(\bl_{T})& = \lambda_{T}C_T(\bl_{T-1})+ 1\label{updateCAFF}
\end{align}

To perform updates (\ref{updateJAFF}) and (\ref{updateCAFF}) we first need to define the new time $T$ forgetting factor $\lambda_T$. To do so, we require a cost function $L_{T}(\bl_{T-1})$ off which to learn, and minimize it with one-step gradient descent 
\begin{equation}
	\lambda_{T} = \lambda_{T-1} - r_\lambda \frac{\partial}{\partial \pmb{\lambda}_{T-1}} L_{T}(\bl_{T-1}),
	\label{eq: lambda update}
\end{equation}
where $r_\lambda$ is a learning rate parameter. 

Consider a real-valued quantity of interest $Y(\pmb{\lambda}_T)$ that is a function of $\pmb{\lambda}_T$ and a sequence of observations $\{Z_i\}_{i = 1,...,T+1}$, such that
\begin{align}
	Y(\bl_T) & = \lambda_T Y(\bl_{T-1}) + Z_{T+1} \label{eq: form Y} \\
	& = \sum_{i=1}^{T+1}\left[\prod_{j=i}^T\lambda_j\right]Z_i.\nonumber
\end{align}
Defining the derivative of $Y(\pmb{\lambda}_T)$ with respect to $\pmb{\lambda}_T$ as
$$
\frac{\partial}{\partial \pmb{\lambda}_T} Y(\pmb{\lambda}_T) = \lim_{\epsilon \rightarrow 0} \frac{1}{\epsilon}\left[Y(\pmb{\lambda}_T + \epsilon) - Y(\pmb{\lambda}_T)\right],
$$
where
$$
\pmb{\lambda}_T + \epsilon = (\lambda_1 + \epsilon, ..., \lambda_T + \epsilon),
$$
it is shown in \cite{cit7} that
\begin{equation}
	\frac{\partial}{\partial \pmb{\lambda}_T} Y(\pmb{\lambda}_T) = \lambda_{T}\frac{\partial}{\partial \pmb{\lambda}_{T-1}} Y(\pmb{\lambda}_{T-1})+ Y(\pmb{\lambda}_{T}).
	\label{eq: sequential gradient formula}
\end{equation}
The update steps (\ref{updateJAFF}) and (\ref{updateCAFF}) are in the form of (\ref{eq: form Y}). This allows us to take advantage of (\ref{eq: sequential gradient formula}), substituting $Y(\pmb{\lambda}_T)$ for $J_{T+1}(f;\bl_{T})$ and $C_{T+1}(\bl_{T})$, resulting in the update formulas
\begin{equation*}
	\frac{\partial}{\partial \pmb{\lambda}_T} J_{T+1}(f;\bl_{T}) = \lambda_{T}\frac{\partial}{\partial \pmb{\lambda}_{T-1}} J_{T}(f;\bl_{T-1}) + J_{T}(f;\bl_{T-1}),
	\label{eq: sequential gradient formula J}
\end{equation*}
and
\begin{equation*}
	\frac{\partial}{\partial \pmb{\lambda}_T} C_{T+1}(\bl_T) = \lambda_{T}\frac{\partial}{\partial \pmb{\lambda}_{T-1}} C_{T}(\bl_{T-1}) + C_{T}(\bl_{T-1}).
	\label{eq: sequential gradient formula C}
\end{equation*}

Computing $\frac{\partial}{\partial \pmb{\lambda}_T}J_{T+1}(f;\bl_T)$ and $\frac{\partial}{\partial \pmb{\lambda}_T}C_{T+1}(\bl_T)$ allows us to compute $\frac{\partial}{\partial \pmb{\lambda}_T}\widehat{S}_{T+1}(f;\bl_T)$ and thus $\frac{\partial}{\partial \pmb{\lambda}_T}L_{T+1}(\bl_T)$ for any cost function dependent on $\widehat{S}_{T+1}(f;\bl_T)$. We are now left with the choice of a suitable cost function to minimize, with the Whittle Likelihood making for an obvious objective function. Including $\mathcal{L}(S_{\widehat{\phi}_{T+1}};\bl_T)$ in the cost function allows us a seamless way to incorporate the performance of the current estimators in the tuning regime. Defining
$L_{T+1}(\bl_T)= - \mathcal{L}_{T+1}(S_{\widehat{\phi}_{T+1}};\bl_T)$, and making the substitution in (\ref{eq: lambda update}) results in the update step
$$
\lambda_{T+1} = \lambda_T + r_\lambda \frac{\partial}{\partial \pmb{\lambda}_T} \mathcal{L}(S_{\widehat{\phi}_{T+1}};\bl_T).
$$
We call this procedure \emph{adaptive forgetting factor Whittle estimation} (AFFWE).
\section{Experiments}

To assess the performance of the presented online algorithms, we start by looking at signals composed of a single major frequency in noise and show how the FFP is able to track changes in this major frequency with a higher degree of fidelity when compared to the periodogram. We then look at a situation where the interest is in estimating the entire SDF as opposed to focusing on a single major frequency. In particular, we study the performance of the FFWE and the AFFWE when applied to AR(2) processes. An advantage of looking at this specific family of SDFs is that we can draw parallels to the first simulation study on a single frequency signal plus noise due to the fact that some AR($2$) processes exhibit strong cyclical behaviour around a single frequency. We also present a simulation study on AR(3) processes, to investigate the effect of more parameters in the model.
\subsection{Single frequency}

Let $\{X_t\}$ be the random process $X_t = \sin(g(t) + \zeta) + \epsilon_t$. Here, $\{\epsilon_t\}$ is a standard Gaussian white noise process, $\zeta \sim U[-\pi, \pi]$, and $g(t)$ is a modulating function such that $g'(t)\in(0, 0.5)$ for all $t$. Its derivative $g'(t)$ is the instantaneous frequency of the process at time $t$. For a process with a single dominant frequency $\eta$, we have $g'(t) = \eta$ and $g(t) = \eta t$ (the constant is absorbed into $\zeta$). The corresponding SDF for such a time series is $S(f;\eta) = \sigma^2 + k\delta(f-\eta)$, where $\delta(\cdot)$ is the Dirac delta function.
We simulate a process with a changing instantaneous frequency by letting
\begin{equation}
\label{changing freq}
g'(t) = (1+0.6\cdot\sin(\gamma t))/4,
\end{equation}
where $\gamma>0$ is a parameter controlling the rate at which the process undergoes concept drift. This results in $$g(t) = \frac{t}{4}+\frac{3}{20\gamma}(1-\cos(\gamma t)).$$

Given an estimate $\hat{S}_T(f)$ for the SDF $S_T(f)$, we define the estimator for the instantaneous frequency $g'(T)$ as
\begin{equation}
	\hat{g}_T = \argmax{f}\{\hat{S}_T(f)\}.
	\label{eq: hat f'}
\end{equation}
Moreover, define $\hat{g}_T(\lambda)$ as the estimator obtained by using the forgetting factor periodogram $\hat{S}_T(f;\lambda)$ in (\ref{eq: hat f'}). Figures \ref{fig: single f example gamma 1} and \ref{fig: single f example gamma 8} show the FFP based estimator when applied to this process, where $\gamma = 1/1000$ and $\gamma = 8/1000$, respectively. Recall that $\lambda = 1$ is equivalent to using a sequentially updated periodogram. 

In Figure \ref{fig: single f example gamma 1} we see how the base periodogram struggles to adapt to the changing form of the underlying data generating process and indeed obscures some of the features, such as the trough in the second half of the time series. On the other hand, the two estimates obtained using forgetting factors result in a higher fidelity to the ground truth. We also note that, depending on the value of $\gamma$ (how quickly the process is undergoing concept drift), the performance of the two forgetting factors differs substantially. In general, we see that a higher value of $\lambda$ (less forgetting) results in the estimate lagging behind the true value. However, $\gamma$ influences how closely the estimate tracks the true value of $g'(T)$. When concept drift is more pronounced (larger $\gamma$), a smaller value for $\lambda$ (more forgetting) produces estimates that better follow the true dynamics. This is expected, since a rapidly changing process naturally benefits from an estimator that adapts quickly by giving less weight to outdated data.

\begin{figure}[H]
\centering
\begin{subfigure}{0.45\linewidth}
	\includegraphics[width=\linewidth]{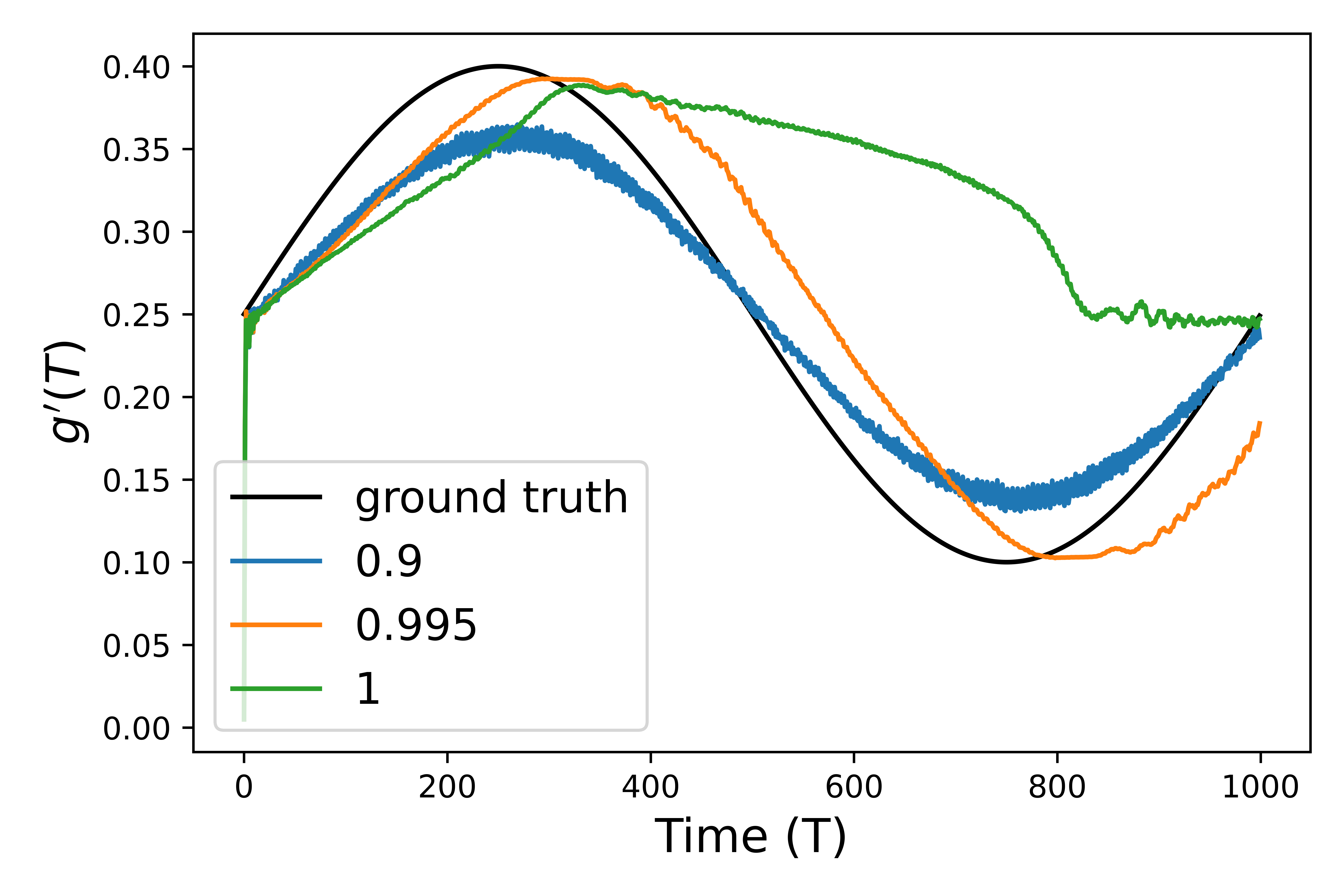}
	\caption{} \label{fig: single f example gamma 1}
\end{subfigure}
\begin{subfigure}{0.45\linewidth}
	\includegraphics[width=\linewidth]{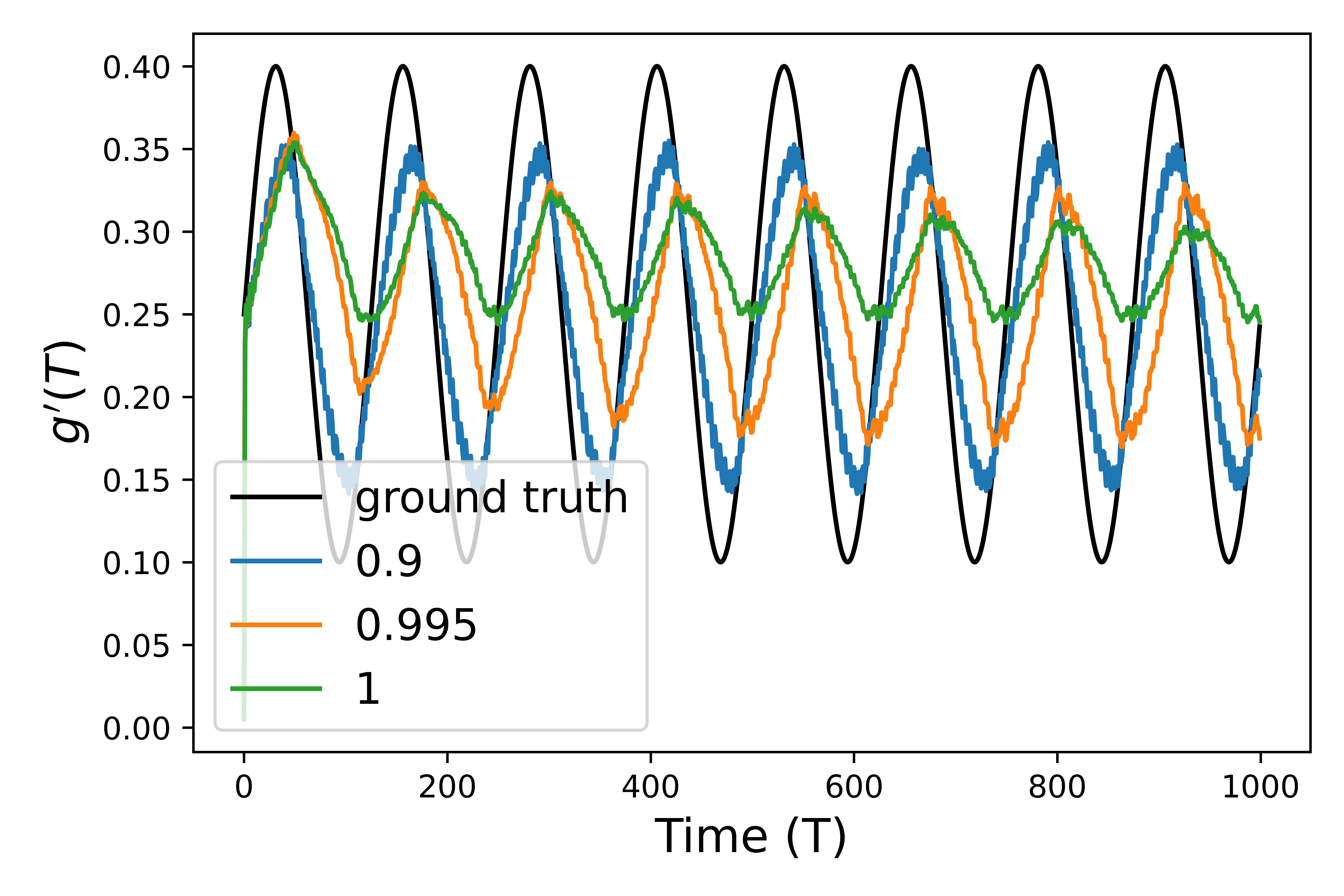}
	\caption{} \label{fig: single f example gamma 8}
\end{subfigure}
\caption{Estimate derived via the FFP for the dominant frequency of a noisy signal composed of a pure sinusoidal and a noise component. The true frequency of the sinusoidal changes over time as shows by the `ground truth' curve. The other three curves represents estimates at three different forgetting factors levels, as shown in the legend. (a) is an example of a slowly changing dominant frequency (low degree of concept drift). (b) is an example of fast changing dominant frequency (high degree of concept drift). The figures were generated using 10000 Monte Carlo simulations.}
\end{figure}

Define the Integrated Squared Error of the estimator $\hat{g}_T(\lambda)$ as 
\begin{equation*}
	\text{ISE}(\{\hat{g}_T(\lambda)\}|g(\cdot)) = \sum_{t = 1}^T \left[\hat{g}_T(\lambda) - g'(T)\right]^2.
\end{equation*}
We will refer to the mean integrated squared error (MISE) as the mean ISE for the estimator across Monte Carlo simulations.

In Figure \ref{fig: single f mse} we compare the MISE of our estimator for a range of forgetting factors as we vary the concept drift parameter $\gamma$. Immediately evident is the fact that we are able to reaffirm the previously stated observation that as $\gamma$ increases, lower values of $\lambda$ give a better performance. Furthermore, note that all forgetting factors estimators perform better than the simple periodogram ($\lambda = 1$).

\begin{figure}[H]
	\centering
	\includegraphics[width=0.6\linewidth]{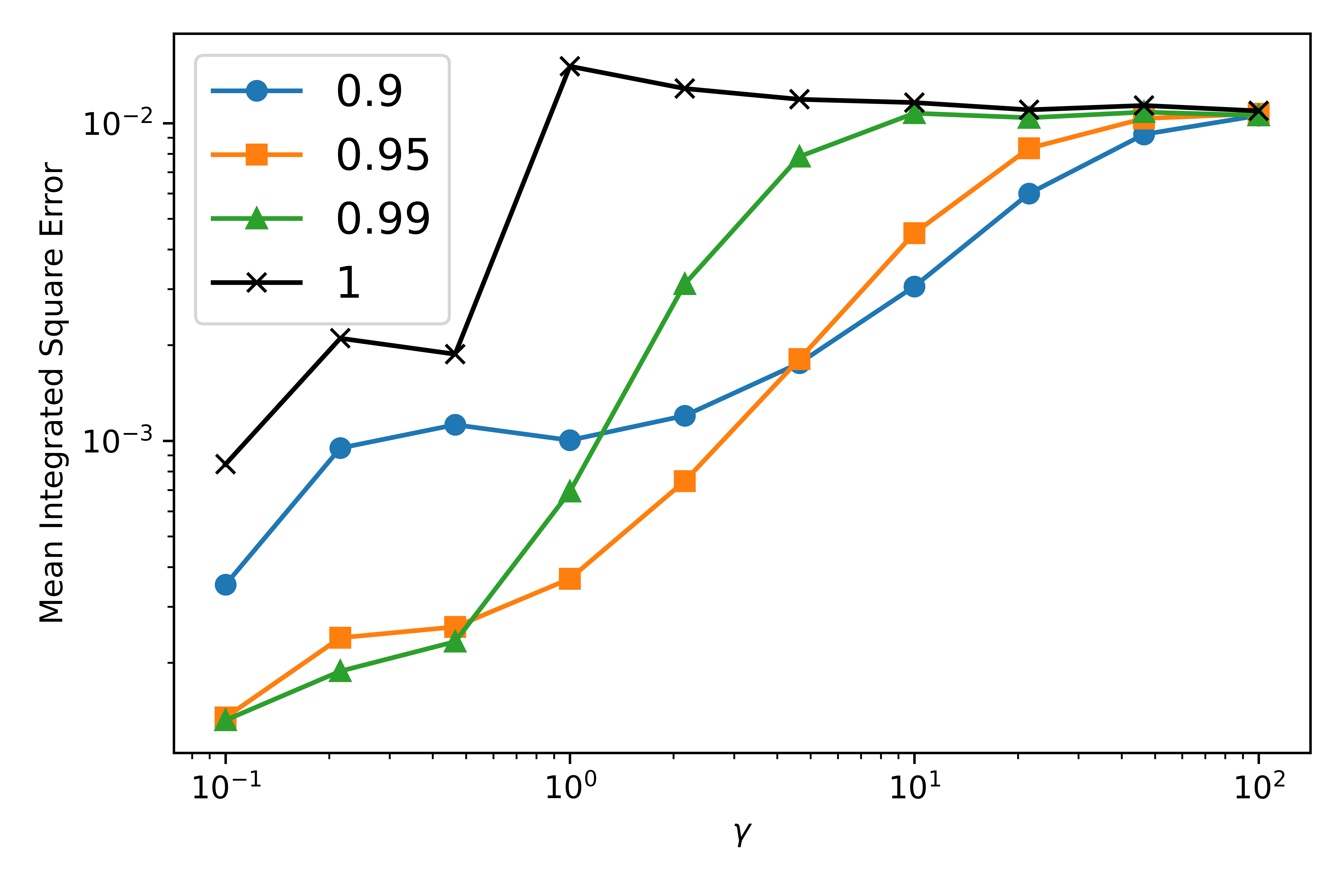}
	\caption{The MISE for the FFP based estimator for the central frequency of a signal, as given in (\ref{changing freq}). Each curve represents a different level of forgetting factor $\lambda$ (see legend). The fact that the curves for each value of $\lambda$ overtake each other at different values of $\gamma$ provides evidence of the how the degree of concept drift inherent to the data generating process has an effect on the best value for the control parameter $\lambda$. Estimates obtained via 1000 Monte Carlo simulations.}
	\label{fig: single f mse}
\end{figure}

\subsection{Autoregressive Processes}

We demonstrate the adaptive forgetting factor approach of Section \ref{sect: Adap FFP} on autoregressive processes. An AR$(p)$ process has the form $\Phi(B)X_t = \epsilon_t$, where $B$ is the backward shift operator $BX_t \equiv X_{t-1}$ and 
$\Phi(z) \equiv 1 - \phi_{1} z - \phi_{2} z^2 - ... - \phi_{p} z^p$
is the \emph{characteristic polynomial} of the AR($p$) process, where $\phi_{1}, \phi_{2}, ..., \phi_{p}$ are constants such that $\phi_{p} \neq 0$, and $\{\epsilon_t\}$ is a Gaussian white noise process with mean $0$ and variance $\sigma^2$ \citep{AR_processes}. Given the AR$(p)$ process satisfies the stationary constraints that the roots of $\Phi(z)$ lie outside of the unit circle, its SDF is
\begin{equation*}
	S(f) = \frac{\sigma^2}{\left|1- \phi_{1} e^{-i 2\pi f} - \phi_{2} e^{-4\pi i f} - ...  -\phi_{p} e^{-2p\pi i f}\right|^2}.
\end{equation*}
In particular, for an AR$(2)$ process whose characteristic polynomial $\Phi(z)$ has complex conjugate roots $z_0$ and $z_0^\ast$, parameterised as $z_0=(1/r)\exp(-i\pi f' t)$, with $0<r<1$ to ensure stationarity, the SDF takes the form \citep{AR_processes} 
\begin{equation*}
	S(f) = \frac{\sigma^2}{(A - 2r\cos(2\pi (f' + f)))(A - 2r\cos(2\pi (f' - f)))},
\end{equation*}
where $A = 1+r^2$. The spectrum will be at its largest when the denominator is at its smallest. For $r$ close to 1 this occurs at $f = \pm f'$. Furthermore, as $r$ approaches 1 from below, this peak at $\pm f'$ becomes more pronounced. In this setting, the process is said to have \textit{pseudo-cyclical} behaviour, as opposed to the deterministic cycle showing up as a sharp spike in the previous set of simulations.
\iffalse
\begin{figure}[H]
	\centering
	\includegraphics[width=0.6\linewidth]{Figures/compareS_AR2.png}
	\caption{The spectrum of an AR(2) process for different pair of parameters $(f', r)$.}
	\label{fig: AR2 examples}
\end{figure}
\fi
Instead of smooth concept drift, we now consider change-points. In Figure \ref{fig: WE example a} and Figure \ref{fig: WE example b} we have the estimates obtained via the Forgetting Factor Whittle Estimator (FFWE) for an AR$(2)$ process undergoing a change-point at $t = 10000$. These were obtained via 1000 Monte Carlo simulations. The black line indicates the ground truth, whereas all the others are Monte Carlo mean estimates for various values of the forgetting factor $\lambda$. Once again, note the periodogram ($\lambda = 1$) struggles to react and adapt to the change-point. Furthermore, there is an apparent trade-off between reaction speed and bias. The estimator for $\lambda = 0.9$ reacts to the change-point significantly faster than the estimator for $\lambda = 0.999$. However, this comes at the cost of increased bias.

In Figures \ref{fig: AWE phi12 a} and \ref{fig: AWE phi12 b} we apply the Adaptive Forgetting Factor Whittle Estimator (AFFWE) to the same data. A key observation is that the learning rate has very little effect on the performance of the estimators. This provides evidence for the Adaptive Whittle Estimator being significantly more robust to hyperparameter misclassification when compared to the Whittle Estimator with fixed forgetting factors.

An interesting behaviour seen in all four figures is a jump in the estimates for $\phi_2$ despite the true parameter not undergoing a change-point. This dynamic comes as a result of the shape taken by the Whittle likelihood function of an AR(2) process as more data is observed. Appendix B provides more details on this phenomenon.

Figures \ref{fig: WE example ar3} and \ref{fig: AWE phi12 ar3} show the algorithms being applied to a AR(3) process. As expected, the variance of our estimates is higher due to increased model complexity.

\begin{figure}[H]
	\centering
	\begin{subfigure}{0.45\linewidth}
		\includegraphics[width=\linewidth]{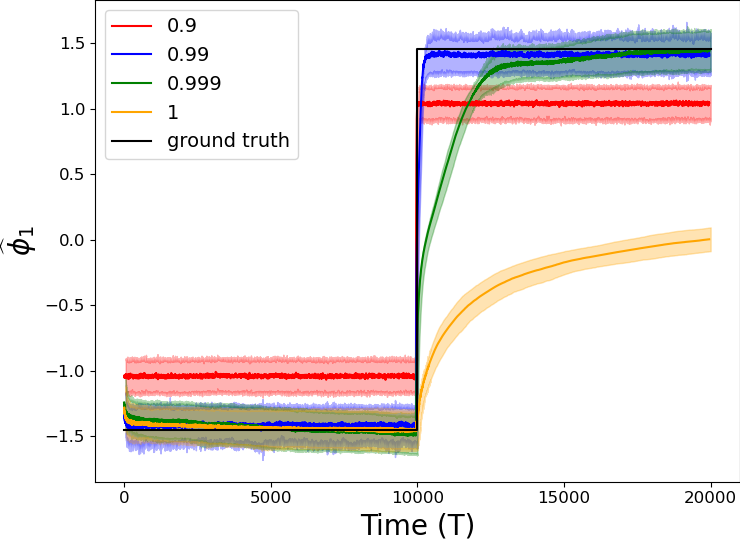}
		\caption{} \label{fig: WE example a}
	\end{subfigure}
	\begin{subfigure}{0.45\linewidth}
		\includegraphics[width=\linewidth]{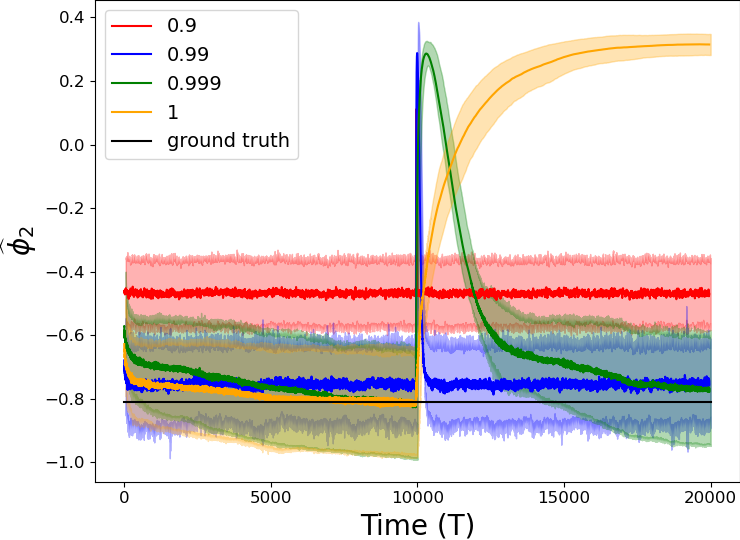}
		\caption{} \label{fig: WE example b}
	\end{subfigure}
	\caption{The curves show the mean of the estimates obtained via Forgetting Factor Whittle Estimation (see legend). The data being analysed was drawn from AR(2) processes undergoing a change-point at $t = 10000$. The black lines represent the ground truth for the AR(2) parameters, $\phi_1$ (a) and $\phi_2$ (b). The shaded area represents the estimated standard deviation. Estimates obtained via 1000 Monte Carlo simulations.}
\end{figure}

\begin{figure}[H]
	\centering
	
	\begin{subfigure}{0.45\linewidth}
		\includegraphics[width=\linewidth]{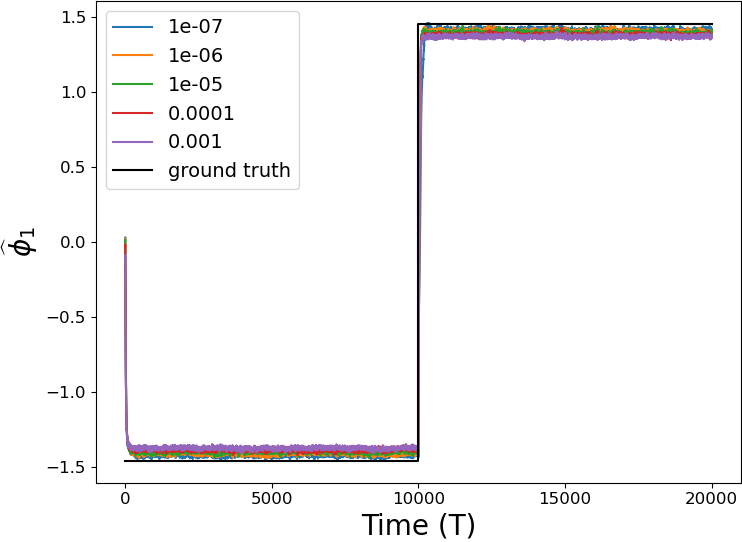}
		\caption{} \label{fig: AWE phi12 a}
	\end{subfigure}
	\begin{subfigure}{0.45\linewidth} 
		\includegraphics[width=\linewidth]{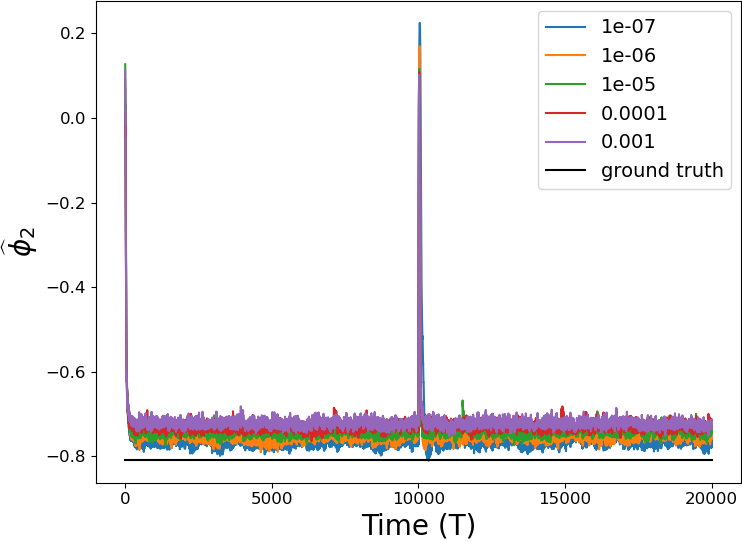}
		\caption{} \label{fig: AWE phi12 b}
	\end{subfigure}
	\caption{The curves show the mean of the estimates obtained via Adaptive Forgetting Factor Whittle Estimation (see legend). The data being analysed was drawn from AR(2) processes undergoing a change-point at $t = 10000$. The black lines represent the ground truth for the AR(2) parameters, $\phi_1$ (a) and $\phi_2$ (b). Estimates obtained via 1000 Monte Carlo simulations.}
	\label{fig: AWE phi12}
\end{figure}

\begin{figure}[H]
	\centering
	\includegraphics[width=0.6\linewidth]{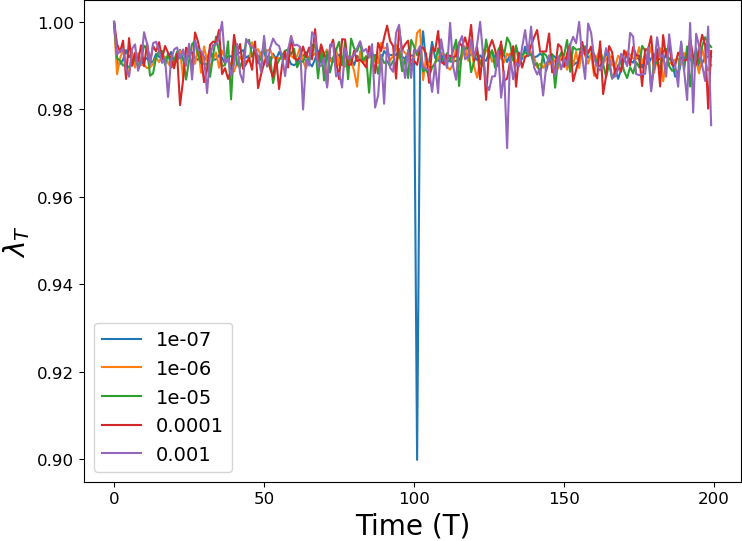}
	\caption{Average forgetting factor at each time point associated with the estimation carried out in Fig. \ref{fig: AWE phi12}. Notice the large spike toward lower values at $t=10000$, when the change-point occurs. The effect of the learning rate for the forgetting factor can be seen in the variance of the adaptive forgetting factor and the rate at which it raises back up after the change-point.}
	\label{fig: AWE eta}
\end{figure}

\begin{figure}[H]
	\centering
	\includegraphics[width=0.32\linewidth]{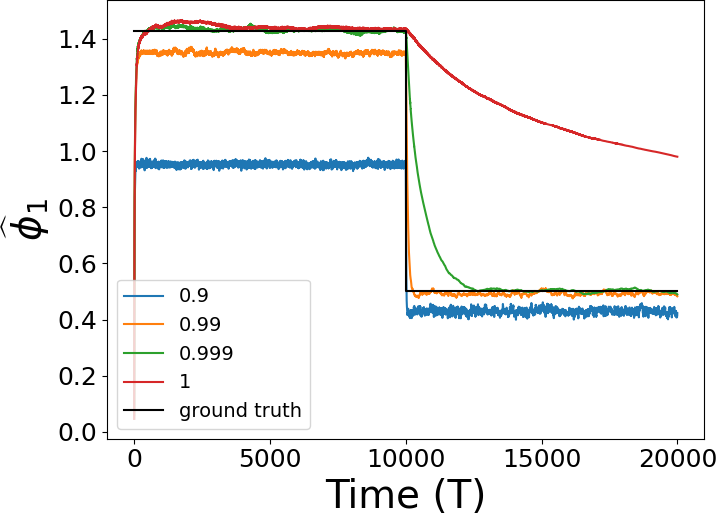}
	\includegraphics[width=0.32\linewidth]{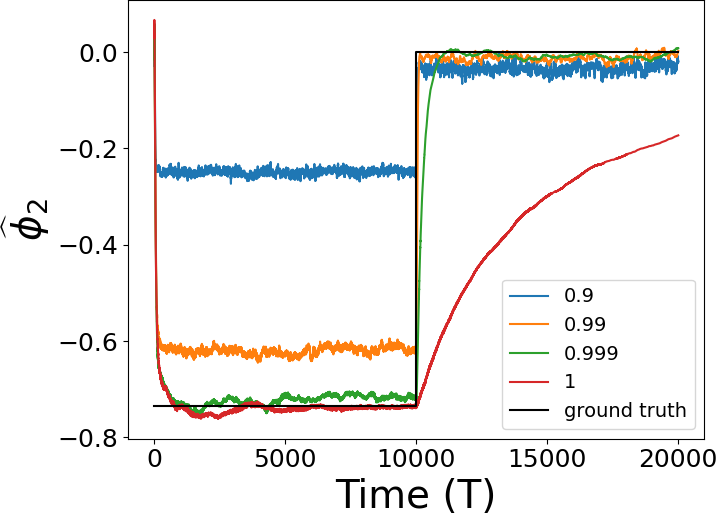}
	\includegraphics[width=0.32\linewidth]{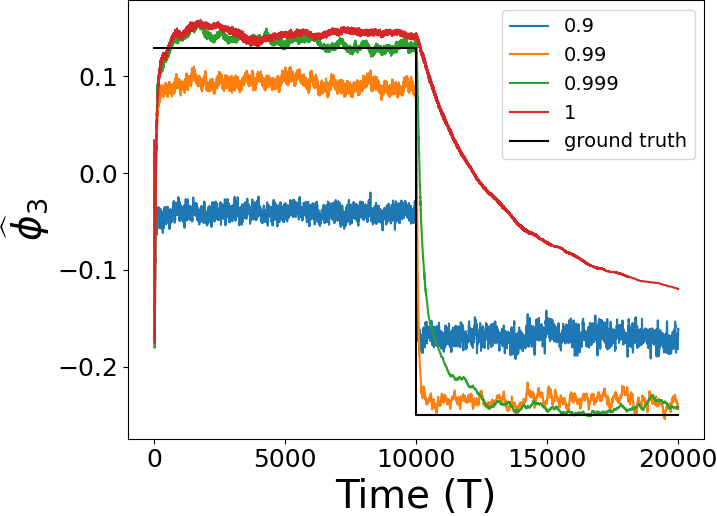}
	\caption{Mean estimates obtained via Forgetting Factor Whittle Estimation (FFWE). The data being analysed was drawn from AR(3) processes undergoing a change-point at $t = 10000$. The black lines represent the ground truth for the AR(3) parameters: $\phi_1$, $\phi_2$  and $\phi_3$. The other lines are estimates with varying levels of forgetting factors, as shown in the legends. Estimates obtained via 1000 Monte Carlo simulations.}
	\label{fig: WE example ar3}
\end{figure}

\begin{figure}[H]
\centering
\includegraphics[width=0.32\linewidth]{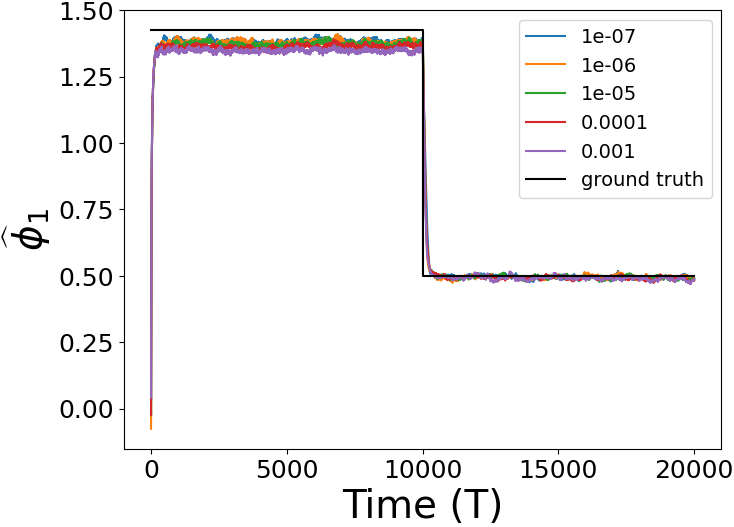}
\includegraphics[width=0.32\linewidth]{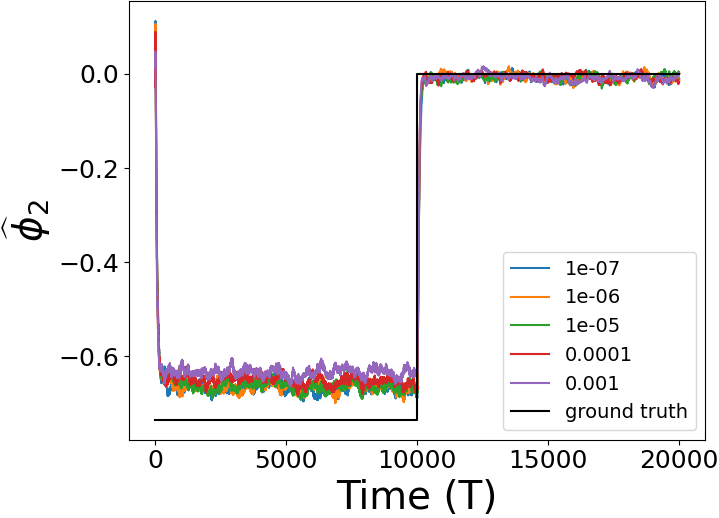}
\includegraphics[width=0.32\linewidth]{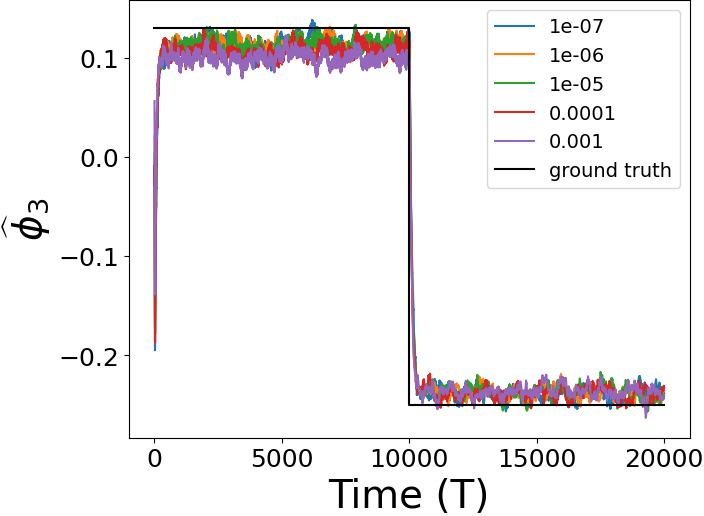}
\caption{Mean estimates obtained via Adaptive Forgetting Factor Whittle Estimation (AFFWE). The data being analysed was drawn from AR(3) processes undergoing a change-point at $t = 10000$. The black lines represent the ground truth for the AR(3) parameters: $\phi_1$, $\phi_2$  and $\phi_3$. The other lines are estimates with varying levels of forgetting factors, as shown in the legends. Estimates obtained via 1000 Monte Carlo simulations.}
\label{fig: AWE phi12 ar3}
\end{figure}

\subsection{Beta Prior}
From simulation studies, even on stationary processes, we find that the average value for the adaptive forgetting factors hovers around $0.9$. Such a value leads to quite an aggressive discarding of past data, even when the data generating process is stationary, causing the variance of the estimator to be higher than necessary. Taking inspiration from the Bayesian regularization literature \citep{chaari2013sparse}, we modify the cost function to include a Beta$(\alpha, 1)$ prior for $\lambda_{T+1}$,
\begin{equation}
	L_{T+1}(\bl_T) = - \mathcal{L}_{T+1}(S_{\hat\phi_{T+1}};\bl_T) - \log P_{\text{Beta}}(\bl_{T}|\alpha, 1).
	\label{eq: lambda cost function beta}
\end{equation}
The reasoning behind this modification is to induce a predisposition in the algorithm towards a forgetting factor of $\alpha/(1+\alpha)$. The choice of $\alpha$ is not trivial. However, we find a value of 1000 or higher results in similar performances while significantly reducing the noise in sequence $\bl_T$.  In Appendix C we provide more examples of how $\alpha$ affects the sequence $\bl_T$.

\begin{figure}[H]
	\centering
	\includegraphics[width=0.6\linewidth]{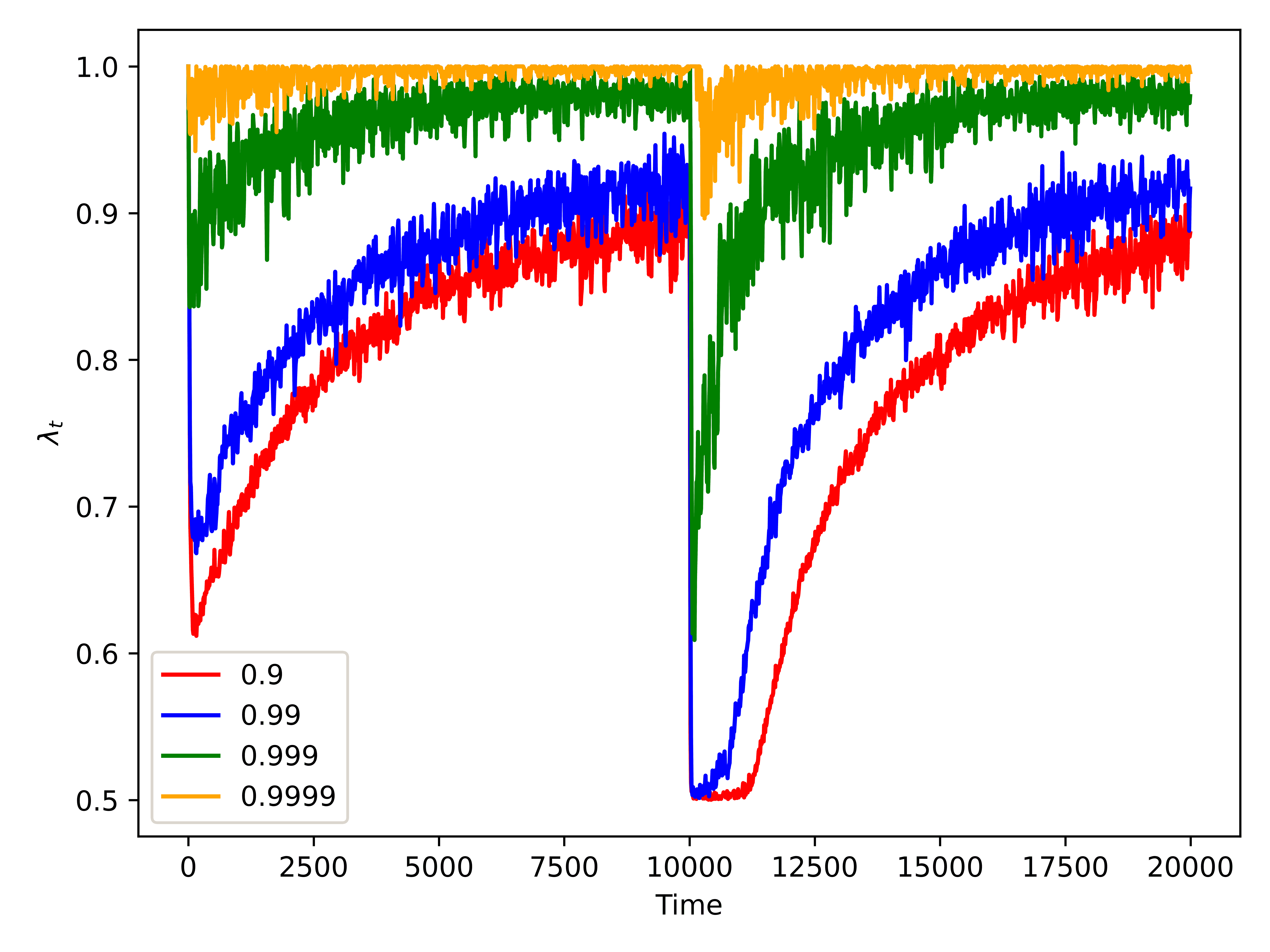}
	\caption{Effect of $\alpha$ in (\ref{eq: lambda cost function beta}) for the estimation of the process used in Figure \ref{fig: AWE phi12}. The process has a change-point at $t = 10000$, where we can see the forgetting factor $\lambda$ drop to suddenly. Note that, based on the value of $\alpha$, the mean $\lambda_t$ levels off at a different value. The legend shows the values for $\alpha/(1+\alpha)$, the mean of the Beta prior on $\lambda_t$.}
	\label{fig: different alpha}
\end{figure}

\section{Oceanic Dataset}
\label{sect: ocean data}

We analyse real-world data from the Global Drifter Program \citep{oceandataset}. Specifically, we build on the work of \cite{ocean_ar2_paper} in which the authors model jointly the latitudinal and longitudinal velocities obtained from instruments known as drifters, which drift freely according to ocean surface flows \citep{ferrari2009ocean}. Those velocities are modelled as the aggregation of two independent complex valued processes, one of which is non-stationary and which we model as a complex-valued AR(1) process. Figure \ref{fig: trajectory} shows the path taken by one of these drifters.
The velocity time series are non-stationary, as they are modulated by oscillations that change in frequency over time. The oscillations are known as inertial oscillations – one of the most ubiquitous and readily observable features of the ocean currents accounting for approximately half of the kinetic energy in the upper ocean \citep{ferrari2009ocean}. Inertial oscillations arise owing to the deviation of the rotating Earth from a purely spherical geometry, together with the appearance of the Coriolis force in the rotating reference frame of an Earth-based observer. The modulation of these oscillations occurs because the drifters are changing latitude.

\begin{figure}[H]
	\centering
	\includegraphics[width=0.6\linewidth]{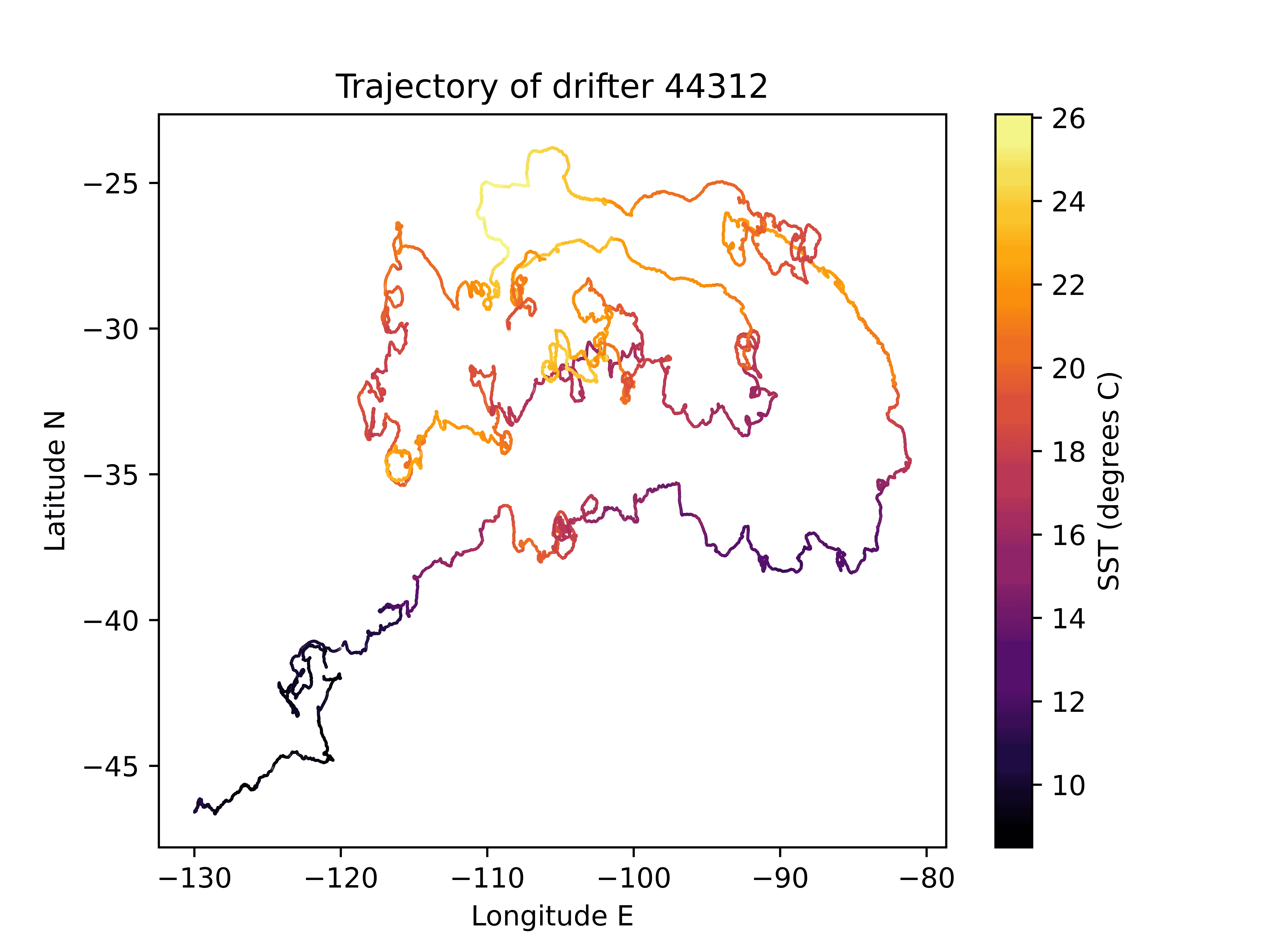}
	\caption{Trajectory of drifter 44312 from the Global Drifter Program.}
	\label{fig: trajectory}
\end{figure}

The second-order structure of complex-valued processes is characterised by both the autocovariance $s_\tau = \text{Cov}[Z_t^*, Z_{t+\tau}]$ and the relation $r_\tau = \text{Cov}[Z_t, Z_{t+\tau}]$ sequences.  Complex valued processes, unlike real-valued, no longer have a spectrum that needs to satisfy Hermitian symmetry, and if the series represents motion in the plane, the positive and negative frequencies represent anti-clockwise and clockwise rotations respectively. Following the classical modelling framework \citep{walden2013rotary} for complex-valued processes, we shall assume that the relation sequence takes the value zero for all lags, and the definition for the SDF in (\ref{eq: spectral density definition}) still applies.

We treat the observed time series as a realisation of the random process $\{Z_t\}_{t\in\mathbb{Z}} = \{X_t + iY_t\}_{t\in\mathbb{Z}}$, where $\{X_t\}$ is the longitudinal velocities and $\{Y_t\}$ is the latitudinal velocities. Complex-valued process $\{Z_t\}$ is conventionally modelled as a modulated AR(1) process
\begin{align*}
	Z_t &= r e^{i\beta_t} Z_{t-1} + \epsilon_t,\\
	Z_0 &\sim \mathcal{CN}\left(0, \frac{\sigma^2}{1-r^2}\right),\\
	\epsilon_t &\sim \mathcal{CN}(0, \sigma^2),
	\label{eq: ocean eq}
\end{align*}
where $\{\beta_t\}$ is a deterministic sequence, which is a function of the \emph{inertial frequency} $\omega_t = -k\sin(\xi_t)$, for some $k\in\mathbb{R}$ where $\xi_t$ is the latitude of the drifter at time $t$. The parameters $r$ and $\sigma^2$ are themselves functions of five parameters, $\{A,\lambda,B,h,\alpha\}$, and if $\omega_t$, and hence $\beta_t$, were fixed, the modulated AR(1) process above would be stationary with SDF
\begin{equation}
	\label{eq:SDF_ocean}
S(f) = \frac{A^2}{(f-\omega)^2 + \lambda^2} + \frac{B}{(\omega^2 + h^2)^\alpha}.
\end{equation}
In practice, $\omega_t$ does vary over time as the drifter moves, but it is known. To estimate the unknown parameters, \cite{ocean_ar2_paper} appropriately demodulate the process using the known $\omega_t$, thereby transforming it into a stationary AR(1) process. They then apply (offline) Whittle estimation.

In this paper, to illustrate the power of our method on real-data, we take a different approach. We apply the FFWE to the raw time series, using (\ref{eq:SDF_ocean}) to estimate all parameters, including online tracking of $\omega_T$ as $T$ increases. We are then able to compare this tracked online estimate with the ground truth.

Figure \ref{fig: S ocean comparison} compares the FFWE at the end of the run against the offline (batch) Whittle estimate on the demodulated process as presented in \cite{ocean_ar2_paper}. While the batch Whittle estimation on the demodulated process more closely matches the periodogram, the performance of the two methods are comparable. In Figure \ref{fig: wf tracking}, we plot our estimates for $\omega_T$ against the trajectory, as a function of $T$. The FFWE is clearly able to track the ground truth without assuming any prior information aside from the assumption of a modulated complex AR(1) process. Furthermore, even a very modest level of forgetting ($\lambda = 0.9999$) gives significant advantages of no forgetting ($\lambda = 1$) in the method's responsiveness to drift.

\begin{figure}[H]
\centering
\includegraphics[width=0.6\linewidth]{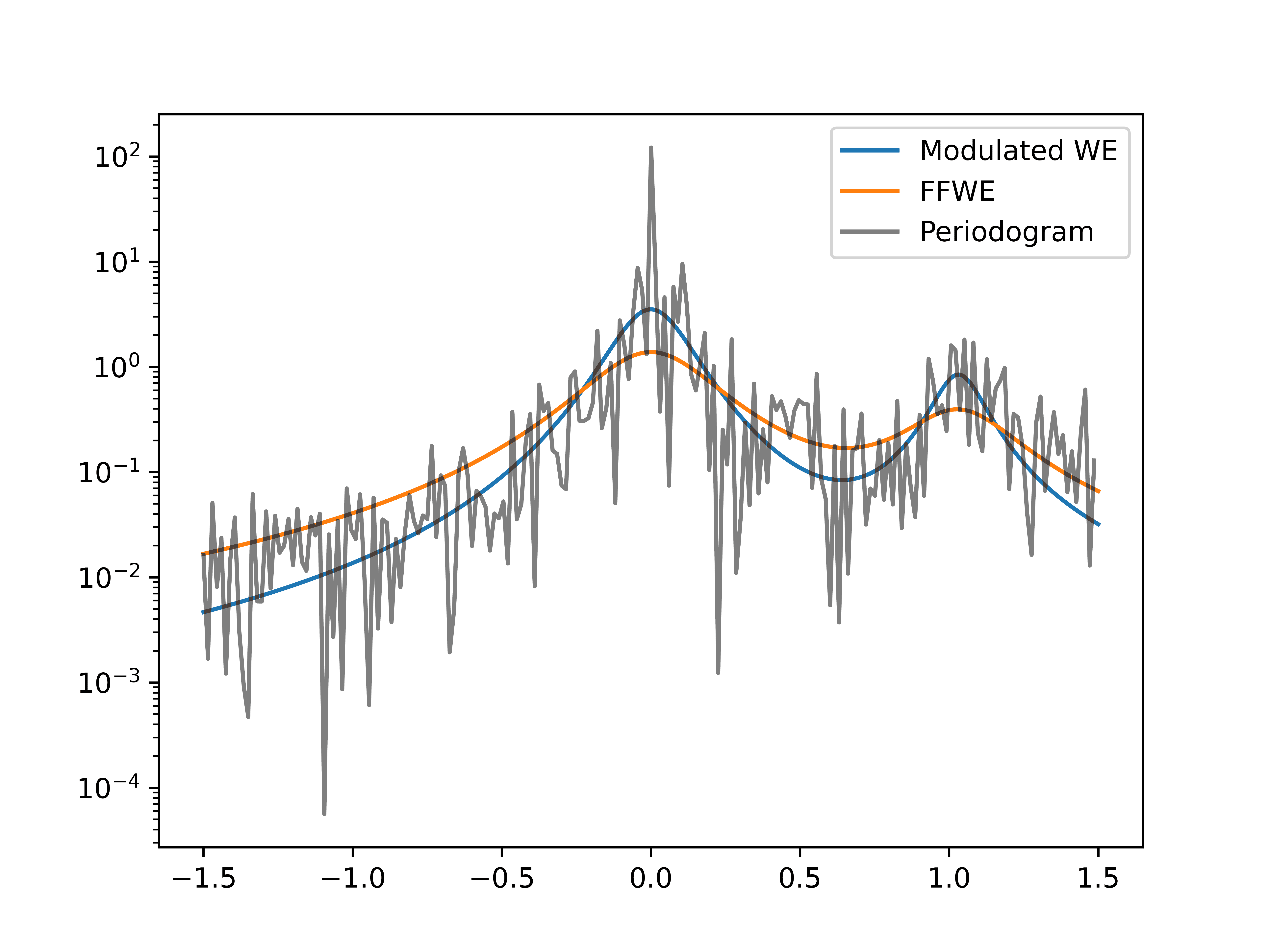}
\caption{Estimated SDFs for the velocity data of drifter 44312 using a periodogram, the FFWE and the demodulation procedure of \cite{ocean_ar2_paper}.}
\label{fig: S ocean comparison}
\end{figure}
\begin{figure}[H]
\centering
\includegraphics[width=0.6\linewidth]{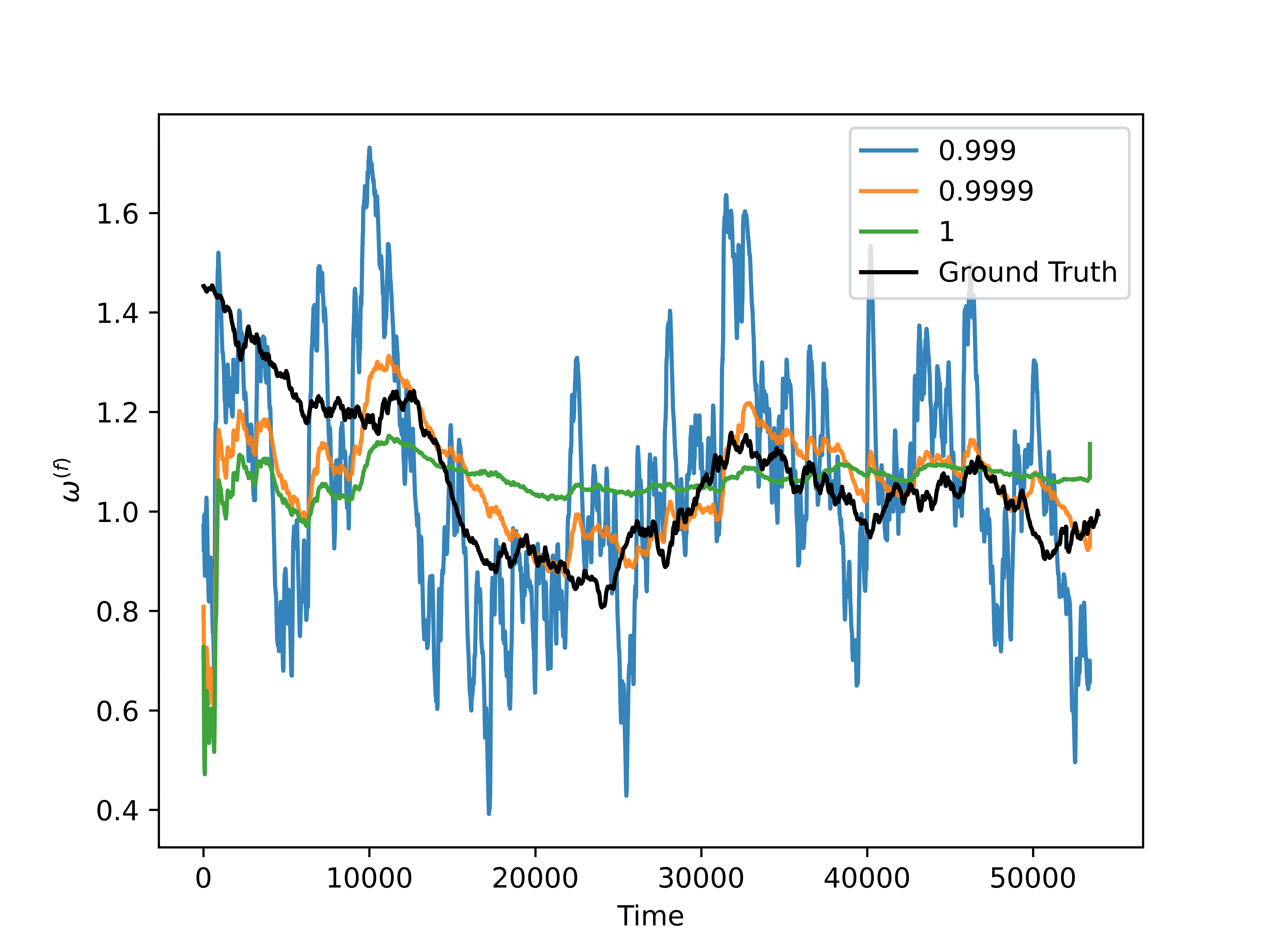}
\caption{Online tracking of $\omega_T$ as a function of $T$ using the FFWE on the velocity data of drifter 44312. The black curve represents the ground truth and the other curves are estimates for different levels of forgetting (see legend).}
\label{fig: wf tracking}
\end{figure}

\section*{Funding}
Shahriar Hasnat Kazi acknowledges funding from the Engineering and Physical Sciences
Research Council (EPSRC), grant number EP/S023151/1. Ed Cohen acknowledges funding
from the EPSRC NeST Programme, grant number EP/X002195/1.

\section*{Disclosure statement}\label{disclosure-statement}

The authors declare no conflict of interest.

\section*{Supplementary Materials}

The oceanographic data used in Section \ref{sect: ocean data}, are openly available at \url{https://doi.org/10.25921/x46c-3620}.

\begin{appendices}
	
	\section{Asymptotic Distribution of the Forgetting Factor Periodogram}
	\label{proof: 1}
	To derive the asymptotic distribution of the FFP, we adapt the proof for the asymptotic distribution of the periodogram \citep{brillinger2001}. 
	
	\begin{assumption}
		\label{assump:FF}
		Let $\bl_T = (\lambda_t)_{t=1,...,T}$ be an increasing sequence of forgetting factors with $\lambda_t\in(0,1)$ for all $t>0$ and $1-\lambda_T^2 = o(T^{-\gamma})$ as $T\rightarrow\infty$, for some $\gamma>1$.
	\end{assumption}	
	
	\begin{definition}\label{def:ffp}
		Let $\{X_t\}_{t\in\mathbb{Z}}$ be a zero-mean second-order stationary process. For a portion of the process $X_1,...,X_T$, define $J_T(f;\bl_{T-1}) = \sum_{t=1}^T h_{T,t} X_t e^{-i2\pi ft}$ and $C_T(\bl_{T-1}) = \sum_{t=1}^T h_{T,t}^2$, where
	   \begin{equation*}
		h_{T,t} = \begin{cases}
			1, & t = T\\
			\prod_{s = t}^{T-1}\lambda_s, & t = 1, ..., T-1\\
			0, &\textnormal{otherwise.}
		\end{cases}
	\end{equation*}
	The FFP is defined as $$\hat S_T(f;\bl_{T-1}) = \frac{1}{C_T(\bl_{T-1}) }\left|J_T(f;\bl_{T-1}) \right|^2.$$
		\end{definition}
    
    \begin{definition}\label{def:kfm} For data taper $\bh_T = (h_{T,t})_{t=1,...,T}$, the $k$-th Fourier moment is defined as
    \begin{equation*}
		H^{(T)}_k(f) = \sum_{t=1}^T [h_{T,t}]^k\exp\{-i2\pi f t\}.
	\end{equation*}
	\end{definition}
	To aid notation, from here on we use $h_T(t)\equiv h_{T,t}$, retaining $t$ as a discrete index in $\mathbb{Z}$.
	
    \begin{lemma}
		Let $h_T(t)$ and $H_k^{(T)}(f)$ be as defined in Definitions \ref{def:ffp} and \ref{def:kfm}, respectively. For any $u_1, ..., u_{k-1}\in \{0, ..., T\}$ there exists finite $K>0$ s.t.
		$$\left|\sum_t h_T(t+u_1)\cdots h_T(t+u_{k-1})h_T(t) \exp\{-i 2\pi f t\} -H^{(T)}_k(f)\right|\leq K\sum_{a=1}^{k-1} |u_a|.$$
		\label{lemma 1}
	\end{lemma}

    \begin{proof}
    	It is true that
  \begin{align*}
  	&\left|\sum_{t=1}^T h_T(t+u_1)\hdots h_T(t+u_{k-1})h_T(t) \exp\{-i2\pi f t\} -H^{(T)}_k(f)\right|\\
  	&\left|\sum_{t=1}^T h_T(t+u_1)\hdots h_T(t+u_{k-1})h_T(t) \exp\{-i2\pi f t\} -\sum_{t=1}^T [h_T(t)]^k\exp\{-i2\pi f t\}\right|\\
  	\leq & \sum_{t=1}^T \left| h_T(t+u_1)\hdots h_T(t+u_{k-1})h_T(t) \exp\{-i2\pi f t\} -[h_T(t)]^k\exp\{-i2\pi f t\}\right|\\
  	= & \sum_{t=1}^T \left| h_T(t+u_1)\hdots h_T(t+u_{k-1})h_T(t) -[h_T(t)]^k\right|\\
  	= & \sum_{t=1}^T \left|\prod_{a = 1}^{k} h_T(t+u_a) - [h_T(t)]^k\right|,\,\,\,\, \text{where $u_k = 0$.}\\
  \end{align*}
Assume that
  \begin{equation}
  	\left| h_T(t+u_1)\hdots h_T(t+u_{n-1})h_T(t+u_n) -[h_T(t)]^n\right| \leq  \sum_{a=1}^n |h_T(t+u_a) - h_T(t)|.
  	\label{eq: induction eq}
  \end{equation}
  Then it follows that
  \begin{align}
  	&\left| h_T(t+u_1)\hdots h_T(t+u_{n})h_T(t+u_{n+1}) -[h_T(t)]^{n+1}\right|\nonumber \\
    &= \left| \left[h_T(t+u_1)\hdots h_T(t+u_{n})\right]h_T(t+u_{n+1}) -[h_T(t)]^{n}h_T(t)\right|\nonumber\\
  		&\leq \max(h_T(t+u_{n+1}), [h_T(t)]^{n})\left[\left|h_T(t+u_1)\hdots h_T(t+u_{n})-[h_T(t)]^{n}\right| + |h_T(t+u_{n+1}) - h_T(t)|\right] \label{step}\\
  		&\leq \max(h_T(t+u_{n+1}), [h_T(t)]^{n})\left[\sum_{a=1}^n |h_T(t+u_a) - h_T(t)| + |h_T(t+u_{n+1}) - h_T(t)|\right]\nonumber\\
  		&\leq \max(h_T(t+u_{n+1}), [h_T(t)]^{n})\left[\sum_{a=1}^{n+1} |h_T(t+u_a) - h_T(t)|\right]\nonumber\\
  		&\leq \sum_{a=1}^{n+1} |h_T(t+u_a) - h_T(t)|,\,\,\,\,\,\,\text{because $h(x)\leq 1$ for all $x$.}\nonumber
  \end{align}
  Step (\ref{step}) comes from the fact that given positive values $a, b, c, d$, we have that
  \begin{align*}
  	|ab - cd| &= |(a-c)b+c(b-d)|\\
  	&\leq \max(b, c)\left|(a-c) + (b-d)\right|\\
  	&\leq \max(b, c)\left[|a-c| + |b-d|\right].
  \end{align*}
  Given (\ref{eq: induction eq}) is true for $n$ implies the same for $n+1$, and that it is trivially true for $n=1$, by induction it is true for all positive integers $n$. Therefore,
  \begin{align*}
  	&\left|\sum_{t=1}^T h_T(t+u_1)\hdots h_T(t+u_{k-1})h_T(t) \exp\{-i2\pi f t\} -H^{(T)}_k(f)\right|\\
  	\leq & \sum_{t=1}^T \left|\prod_{a = 1}^{k} h_T(t+u_a) - [h_T(t)]^k\right|\\
  	\leq & \sum_{t=1}^T \sum_{a=1}^k |h_T(t+u_a) - h_T(t)|.
  \end{align*}
  Suppose for convenience $u_a\geq0$, (other cases are handled similarly), then
  \begin{align}
  	&\left|\sum_{t=1}^T h_T(t+u_1)\hdots h_T(t+u_{k-1})h_T(t) \exp\{-i2\pi f t\} -H^{(T)}_k(f)\right|\nonumber\\
  	\leq & \sum_{t=1}^T \sum_{a=1}^k |h_T(t+u_a) - h_T(t)|\nonumber\\
  	= & \sum_{t=1}^T \sum_{a=1}^k |(h_T(t+u_a) - h_T(t+u_a-1)) + (h_T(t+u_a-1) - h_T(t+u_a-2)) \nonumber\\
    &\qquad\qquad\qquad\qquad+ \cdots + (h_T(t+1) - h_T(t))|\nonumber\\
  	= & \sum_{a=1}^{k-1} \sum_{t=1}^T \left|\sum_{b = 1}^{u_a} h_T(t+b) - h_T(t+b-1)\right|\nonumber\\
  	\leq & \sum_{a=1}^{k-1} \sum_{t=1}^T \sum_{b = 1}^{u_a} |h_T(t+b) - h_T(t+b-1)|\nonumber\\
  	= & \sum_{a=1}^{k-1} \sum_{b = 1}^{u_a} \left(1 + \sum_{t=1}^{T-b} h_T(t+b) - h_T(t+b-1)\right)\nonumber\\
  	= & \sum_{a=1}^{k-1} \sum_{b = 1}^{u_a} \left(2 - h_T(b)\right)\nonumber\\
  	\leq & K\sum_{a=1}^{k-1} |u_a|\label{step2}
  \end{align}
  for some finite $K$. The step for (\ref{step2}) is given by the fact that $h(x) \in [0, 1]$ for all $x$, thus $$\sum_{b = 1}^{u_a} \left(2 - h_T(b)\right) < \sum_{b = 1}^{u_a} 2 = 2 (u_a-1).$$
  Furthermore, with
  \begin{equation*}
  	h_T(t+b) - h_T(t+b-1) = \begin{cases}
  		-1, & t+b = T+1\\
  		h_T(t+b) - h_T(t+b-1), & t+b = 1, ..., T\\
  		0, &\textnormal{otherwise,}
  	\end{cases}
  \end{equation*}
  it follows that
  \begin{align*}
  	\sum_{t=1}^T |h_T(t+b) - h_T(t+b-1)| &= |-1| + \sum_{t=1}^{T-b} h_T(t+b) - h_T(t+b-1)\\
  	&= 1 + (h_T(T)-h_T(T-1)) + (h_T(T-1)-h_T(T-2))\\
    &\qquad\qquad+ \cdots + (h_T(b+1)-h_T(b))\\
  	&= 1 + (h_T(T)-h_T(b))\\
  	&= 2 - h_T(b),
  \end{align*}
  thus completing the proof.
  \end{proof}
	
    \begin{assumption}
		\label{Asmp 1}
Random process	$\{X_t\}_{t\in\mathbb{Z}}$ is zero mean and strictly stationary, meaning that $\cum(X_{t+u_1}, ..., X_{t+u_k}) = \cum(X_{u_1}, ..., X_{k})$ for all $k>0$, for all $u_1,...,u_k\in\mathbb{Z}$ and for all $t\in\mathbb{Z}$. Furthermore, denoting
	\begin{equation*}
	c_k(u_1, ..., u_{k-1}) \equiv \cum(X_{u_1}, ..., X_{u_{k-1}}, X_{0}),
    \end{equation*}
the following mixing condition holds,
        \begin{equation*}
		\sum_{u_1 = -\infty}^\infty\sum_{u_2 = -\infty}^\infty...\sum_{u_{k-1} = -\infty}^\infty |c_k(u_1, u_2, ..., u_{k-1})| < \infty.
	  \end{equation*}
    \end{assumption}
    
    \begin{definition}
    		We define the $k^{th}$ order cumulant spectrum, $\mathcal{C}_k(f_1, f_2, ..., f_{k-1})$ as
    	\begin{equation*}
    		\mathcal{C}_k(f_1, ..., f_{k-1}) = (2\pi)^{-k+1} \sum_{u_1, ..., u_{k-1} = -\infty}^{\infty} c_k(u_1, ..., u_{k-1}) \exp\left\{ -i 2\pi\sum_{j=1}^{k-1}u_j f_j\right\}.
    	\end{equation*}
    	\end{definition}
    
	\begin{lemma}
    Let $\{X_t\}$ be a random process such that Assumption (\ref{Asmp 1}) holds, then
		\begin{equation*}
			\mathcal{C}_k(f_1, ..., f_{k-1}) =  (2\pi)^{-k+1}\sum_{u_1 = -T}^T\hdots\sum_{u_k = -T}^T c_k(u_1, ..., u_{k-1}) \exp\left\{ -i 2\pi \sum_{j=1}^{k-1}u_j f_j\right\} + o(1).
		\end{equation*}
		\label{lemma: f error}
	\end{lemma}
\begin{proof}
The mixing condition in Assumption \ref{Asmp 1} implies that there exists an $L>0$ such that
$$
\lim_{T\rightarrow\infty} \sum_{u_1 = -T}^T\sum_{u_2 = -T}^T...\sum_{u_{k-1} = -T}^T  |c_k(u_1, u_2, ..., u_{k-1})| = L.
$$
This further implies that
$$\lim_{T\rightarrow\infty} \sum_{u_1, u_2, ..., u_{k-1} \notin\{-T, ..., T\}}  |c_k(u_1, u_2, ..., u_{k-1})| = 0.
$$
It now follows that
    \begin{align*}
        &\left|\mathcal{C}_k(f_1, ..., f_{k-1}) - (2\pi)^{-k+1}\sum_{u_1 = -T}^T\hdots\sum_{u_k = -T}^T c_k(u_1, ..., u_{k-1}) \exp\left\{ -i 2\pi\sum_{j=1}^{k-1}u_j f_j\right\}\right|\\
        =& (2\pi)^{-k+1} \left| \sum_{u_1, u_2, ..., u_{k-1} \notin\{-T, ..., T\}}  c_k(u_1, ..., u_{k-1}) \exp\left\{ -i2\pi \sum_{j=1}^{k-1}u_j f_j\right\} \right|\\
        \leq & (2\pi)^{-k+1} \sum_{u_1, u_2, ..., u_{k-1} \notin\{-T, ..., T\}}  |c_k(u_1, ..., u_{k-1})|
        \longrightarrow \, 0 \quad \textnormal{as}\quad T\longrightarrow \infty.
    \end{align*}
    Thus proving Lemma \ref{lemma: f error}.
    \end{proof}

	\begin{lemma}
Let $X_1,...,X_T$ be a portion of a random process $\{X_t\}$ that follows Assumption \ref{Asmp 1}. Let $\bl_T = (\lambda_t)_{t=1,...,T}$ be a sequence of forgetting factors, with taper $(h_{T,t})_{t=1,...,T}$ and $J_T(f;\bl_{T-1})$ as given in Definition \ref{def:ffp}, and let the $k$th Fourier moment $H_k^{(T)}(f)$ be as defined in Definition \ref{def:kfm}. For all $k>0$, it holds that
		\begin{align*}
			\cum\{&J_T(f_1;\bl_{T-1}), \hdots, J_T(f_k;\bl_{T-1})\} = \\ &H_k^{(T)}\left(\sum_{j=1}^k f_j\right)\sum_{u_1 = -T}^T\hdots\sum_{u_k = -T}^T \exp\left\{-i2\pi\sum_{j=1}^{k-1}f_j u_j\right\} c(u_1, ..., u_{k-1}) + \epsilon_T,
		\end{align*}
		where $\epsilon_T$ is $o(T)$.
%		\begin{equation}
%			|\epsilon_T|\leq K \sum_{u_1 = -T}^T\hdots\sum_{u_{k-1} = -T}^T (|u_1|+\hdots+|u_{k-1}|)c(u_1, ..., u_{k-1}).
%		\end{equation}
		\label{Lemma 2}
	\end{lemma}
	\begin{proof}
	The cumulant in Lemma \ref{Lemma 2} has the form
	\begin{align*}
		\sum_{t_1}...\sum_{t_k} & h_{T}(t-1)...h_{T}(t_k)\exp\left\{-i2\pi\sum_{j=1}^k f_j t_j\right\} c_k(t_1 - t_k, ..., t_{k-1} - t_k)\\
		& = \sum_{u_1}...\sum_{u_{k-1}}\exp\left\{-i\sum_{j=1}^{k-1} f_j u_j\right\} c_k(u_1, ..., u_{k-1})\\
        &\qquad\qquad\times\sum_t h_{T}(t+u_1)...h_{T}(t+u_{k-1})h_{T}(t)\exp\left\{-i2\pi\sum_{j=1}^k f_j t\right\}\\
		& = \sum_{u_1}...\sum_{u_{k-1}}\exp\left\{-i2\pi\sum_{j=1}^{k-1} f_j u_j\right\} c_k(u_1, ..., u_{k-1})\left[H_k^{(T)}\left(\sum_{j=1}^k f_j\right) + \epsilon_T'\right],
	\end{align*}  
	where $|\epsilon_T'| \leq K\sum_{a=1}^{k-1} |u_a|$ from Lemma \ref{lemma 1}. Expanding the last bracket leads to the expression in Lemma \ref{Lemma 2} where
    \begin{align*}
        |\epsilon_T| &= \left|\sum_{u_1}...\sum_{u_{k-1}}\exp\left\{-i\sum_{j=1}^{k-1}f_j u_j\right\} c_k(u_1, ..., u_{k-1})\epsilon_T'\right|\\
        &\leq \sum_{u_1}...\sum_{u_{k-1}}\left|\exp\left\{-i\sum_{j=1}^{k-1}f_j u_j\right\} c_k(u_1, ..., u_{k-1}) \epsilon_T'\right|\\
        &\leq \sum_{u_1}...\sum_{u_{k-1}}\left|\exp\left\{-i\sum_{j=1}^{k-1}f_j u_j\right\} c_k(u_1, ..., u_{k-1}) \sum_{a=1}^{k-1} |u_a|\right|\\
        &\leq \sum_{u_1}...\sum_{u_{k-1}}c_k(u_1, ..., u_{k-1}) \sum_{a=1}^{k-1} |u_a|.
    \end{align*}
   Furthermore,
	\begin{equation*}
		T^{-1}|\epsilon_T|\leq K \sum_{u_1 = -T}^T\hdots\sum_{u_{k-1} = -T}^T T^{-1}\sum_{a=1}^{k-1} |u_a| c(u_1, ..., u_{k-1}).
	\end{equation*}
    Note that $T^{-1}\sum_{a=1}^{k-1} |u_a|\longrightarrow0$ as $T\longrightarrow\infty$. Together with assumption (\ref{Asmp 1}), by the dominated convergence theorem, it implies that $T^{-1}|\epsilon_T|\longrightarrow0$ as $T\longrightarrow\infty$.
    \end{proof}

	\begin{proposition}
	Let $\bl_T = (\lambda_t)_{t=1,...,T}$ be a sequence of forgetting factors, with taper $(h_{T,t})_{t=1,...,T}$ as given in Definition \ref{def:ffp}, and the $k$th Fourier moment $H_k^{(T)}(f)$ as defined in Definition \ref{def:kfm}. For a portion $X_1,...,X_T$ of a random process $\{X_t\}$ that follows Assumption \ref{Asmp 1}, it holds that
		\begin{equation*}
			\cum\{J_T(f_1;\bl_{T-1}), \hdots, J_T(f_k;\bl_{T-1})\} = H_k^{(T)}\left(\sum_{j=1}^k f_j\right)\mathcal{C}_k(f_1, ..., f_{k-1}) + o(T).
		\end{equation*}
		\label{Theorem Final}
	\end{proposition}
\begin{proof}
	\begin{align*}
		\cum\{&J_T(f_1;\bl_{T-1}), \hdots, J_T(f_k;\bl_{T-1})\}\\ &= H_k^{(T)}\left(\sum_{j=1}^k f_j\right)\sum_{u_1 = -T}^T\hdots\sum_{u_k = -T}^T \exp\left\{-i2\pi\sum_{j=1}^{k-1} f_j u_j\right\} c(u_1, ..., u_{k-1}) + \epsilon_T\\
		& = H_k^{(T)}\left(\sum_{j=1}^kf_j\right)\sum_{u_1 = -T}^T\hdots\sum_{u_k = -T}^T \exp\left\{-i2\pi\sum_{j=1}^{k-1} f_j u_j\right\} c(u_1, ..., u_{k-1}) + o(T)\\
		& =  H_k^{(T)}\left(\sum_{j=1}^{k}f_j\right) \mathcal{C}_k(f_1, ..., f_{k-1}) + o(T)\quad \textnormal{From Lemma \ref{lemma: f error}.}
	\end{align*}
	\end{proof}

    \begin{theorem}
		\label{thm: limit J}
		Let $\bl_T = (\lambda_t)_{t=1,...,T}$ be a sequence of forgetting factors, with taper $(h_{T,t})_{t=1,...,T}$, that follows Assumption \ref{assump:FF}, and $X_1,...,X_T$ be a portion of a random process $\{X_t\}$ that follows Assumption \ref{Asmp 1} and has SDF $S(f)$. Then, for any $k>0$ and any set of distinct frequencies $f_1,...,f_K$, where $f_i\notin \{-1/2,0,1/2\}$ for all $i=1,...,K$, then $\left[C_T(\bl_{T-1})\right]^{-1/2} J_T(f_i;\bl_{T-1})$, $i=1,...,K$  are asymptotically independent
		$\mathcal{CN}(0, S(f_i))$ random variables as $T\rightarrow\infty$.
	\end{theorem}

\begin{proof}
Without loss of generality, we consider frequencies $f_i\in(0,1/2)$ for all $i=1,...,K$. To show joint Gaussianity, it is sufficient to show that $$\cum\left\{\left[C_T(\bl_{T-1})\right]^{-1/2}J_T(f_{i_1};\bl_{T-1}), \hdots, \left[C_T(\bl_{T-1})\right]^{-1/2}J_T(f_{i_k};\bl_{T-1})\right\} \rightarrow 0,\qquad i_1,...,i_k = 1,...,K$$
as $T\rightarrow\infty$ for all $k>2$, and that the first and second order moments are as stated. Proceeding in this way, for $k>2$
	\begin{multline*}
	\cum\left\{\left[C_T(\bl_{T-1})\right]^{-1/2}J_T(f_{i_1};\bl_{T-1}), \hdots, \left[C_T(\bl_{T-1})\right]^{-1/2}J_T(f_{i_k};\bl_{T-1})\right\} \\
		 = \left[C_T(\bl_{T-1})\right]^{-k/2}\left(H_k^{(T)}\left(\sum_{j=1}^k f_{i_j}\right)\mathcal{C}_k(f_{i_1}, ..., f_{i_{k-1}}) + o(T)\right)\\
		\leq T\left[C_T(\bl_{T-1})\right]^{-k/2}\mathcal{C}_k(f_{i_1}, ..., f_{i_{k-1}}) + o\left(T\left[C_T(\bl_{T-1})\right]^{-k/2}\right),
	\end{multline*}
    where, directly from Definition \ref{def:kfm}, we have $\left|H_k^{(T)}\left(\sum_{j=1}^k f_{i_j}\right)\right| \leq T$.
    Under Assumption \ref{def:ffp} we have $1-\lambda_T^2 = o(T^{-\gamma})$. This means that for all $\epsilon>0$ there exists $T_\epsilon$ such that for all $T>T_\epsilon$ we have $\lambda^2_T>1-\epsilon T^{-\gamma}$. Fix $\epsilon>0$ and the corresponding $T_\epsilon$. Then for all $T> T_\epsilon$,
    \begin{align*}
   C_T(\bl_{T-1} &= C_{T_\epsilon}(\bl_{T-1}) [h_T(T_\epsilon)]^2 + \sum_{t = T_\epsilon+1}^T [h_T(t)]^2\\
        &\geq \sum_{t = T_\epsilon+1}^T [h_T(t)]^2\\
        &= \sum_{t = T_\epsilon+1}^T \prod_{b=t}^{T-1}\lambda^2_b\\
        &\geq \sum_{t = T_\epsilon+1}^T \prod_{b=t}^{T-1} 1-\epsilon b^{-\gamma}\\
        &\geq \sum_{t = T_\epsilon+1}^T \left[1-\epsilon\sum_{b=t}^{T-1} b^{-\gamma}\right]\\
        &\geq \sum_{t = T_\epsilon+1}^T \left[1-\epsilon\zeta(\gamma)\right]\qquad \qquad \text{where $\zeta(\cdot)$ is the Riemann zeta function}\\
        &= (T-T_\epsilon) \left[1-\epsilon\zeta(\gamma)\right]\\
        &= O(T).
    \end{align*}
    Therefore, $\left[C_T(\bl_{T-1})\right]^{-1} = O(T^{-1})$. Here, we used the fact that the Riemann zeta function $\zeta(\gamma)$ is finite for $\gamma>1$ and that
    \begin{align*}
        \prod_{i=1}^n (1-a_i) &= (1-a_1)\left[\prod_{i=2}^n (1-a_i)\right]\\
        &= \left[\prod_{i=2}^n (1-a_i)\right] - a_1\left[\prod_{i=2}^n (1-a_i)\right]\\
        &\geq \left[\prod_{i=2}^n (1-a_i)\right] - a_i,
    \end{align*}
    if $\left[\prod_{i=2}^n (1-a_i)\right] < 1$. This condition is satisfied in our case because each term in the product corresponds to the square of a forgetting factor term, $\lambda_b^2$. All forgetting factors $\lambda_b$ are positive and bounded by $1$, so the same applies to any product of them. By using the previous expression repeatedly, by induction, we see that
    $$ \prod_{i=1}^n (1-a_i) \geq 1-\sum_{i=1}^n a_i.$$
   Therefore, for $k>2$,
    \begin{multline*}
        \cum\left\{\left[C_T(\bl_{T-1})\right]^{-1/2}J_T(f_{i_1};\bl_{T-1}), \hdots, \left[C_T(\bl_{T-1})\right]^{-1/2}J_T(f_{i_k};\bl_{T-1})\right\}  \\ \leq  T\left[C_T(\bl_{T-1})\right]^{-k/2}\mathcal{C}_k(f_{i_1}, ..., f_{i_{k-1}}) + o\left(T\left[C_T(\bl_{T-1})\right]^{-k/2}\right)\\
        \longrightarrow 0\ \ \ \ \ \ \ \textnormal{ as }\  \left[C_T(\bl_{T-1})\right]^{-1} = O(T^{-1}).
    \end{multline*}
 All joint cumulants of order $k>2$ going to 0 implies the random variables $\left[C_T(\bl_{T-1})\right]^{-1/2} J_T(f_i;\bl_{T-1})$, $i=1,...,K$ are jointly Gaussian. We next consider the $k=2$ order cumulants. It is true that
   \begin{multline*}
   	\cum\left\{\left[C_T(\bl_{T-1})\right]^{-1/2}J_T(f_i;\bl_{T-1}),\left[C_T(\bl_{T-1})\right]^{-1/2}J_T(f_j;\bl_{T-1})\right\} \\
   	= \left[C_T(\bl_{T-1})\right]^{-1}\cum\left\{J_T(f_i;\bl_{T-1})J_T(f_j;\bl_{T-1})\right\}\\
   	= \left[C_T(\bl_{T-1})\right]^{-1}\left(H_2^{(T)}\left(f_i + f_j\right)\mathcal{C}_2(f_i) + o(T)\right).
   \end{multline*}
 For $f_i,f_j \in (0,1/2)$, then $f_1 + f_2 \neq 0 \text{ (mod }1)$. Fixing $\epsilon>0$ and the corresponding $T_\epsilon$ such that for all $T>T_\epsilon$ we have $\lambda^2_T>1-\epsilon T^{-\gamma}$, then
 \begin{align*}
 	\left|\left[\sum_{t=1}^T\exp{(2\pi i ft)}\right] - H_2^{(T)}(f)\right| &= \left|\sum_{i=1}^T (1-\lambda_t^2)\exp{(2\pi i ft)}\right|\\
 	&\leq \sum_{i=1}^T (1-\lambda_t^2)\\
 	&\leq \sum_{t=1}^{T_\epsilon} (1-\lambda_t^2) + \sum_{b=T_\epsilon}^T \epsilon b^{-\gamma}\\
 	&\leq T_\epsilon + \sum_{b=T_\epsilon}^T \epsilon b^{-\gamma}\\
 	&\leq T_\epsilon + \epsilon\zeta(\gamma),
 \end{align*}
 where $\zeta(\cdot)$ is the again the Riemann zeta function. Therefore, 
 \begin{equation*}
 	|H_2^{(T)}(f)| \leq \left|\sum_{t=1}^T\exp{(2\pi i ft)}\right| + T_\epsilon + \epsilon\zeta(\gamma) = \frac{\sin(T\pi f)}{\sin(\pi f)}+ T_\epsilon + \epsilon\zeta(\gamma).
 \end{equation*}
 As the right-hand side is bounded for all values of $T$ if $f\neq 0$, if follows that
 \begin{equation*}
 	\left[C_T(\bl_{T-1})\right]^{-1}\left(H_2^{(T)}\left(f_1 + f_2\right)\mathcal{C}_2(f_1) + o(T)\right) \longrightarrow 0 \,\, \text{as} \,\, T\longrightarrow \infty.
 \end{equation*}
Together with asymptotic Gaussianity, this also shows that $\left[C_T(\bl_{T-1})\right]^{-1/2} J_T(f_i;\bl_{T-1})$, $i=1,...,K$ are also asymptotically independent.
   
   We are left with showing that the first and second order moments of $\left[C_T(\bl_{T-1})\right]^{-1/2}J_T(f_i;\bl_{T-1})$ are as stated. The first order moment is immediate from the fact that $\{X_t\}_{t\in\mathbb{Z}}$ is zero mean. We note that for any $f_i, f_j\in(-1/2,1/2]$ such that $f_i + f_j \equiv 0 \text{ (mod }1)$, then $H_2^{(T)}(f_i + f_j) = H_2^{(T)}(0) = C_T(\bl_{T-1})$, and therefore asymptotically
$$
   	\cum\left\{\left[C_T(\bl_{T-1})\right]^{-1/2}J_T(f_i;\bl_{T-1}),\left[C_T(\bl_{T-1})\right]^{-1/2}J_T(f_j;\bl_{T-1})\right\} \longrightarrow \mathcal{C}_2(f_i) = S(f_i).
 $$
Note that $\left[C_T(\bl_{T-1})\right]^{-1/2}J_T(f_i;\bl_{T-1})$ and $\left[C_T(\bl_{T-1})\right]^{-1/2}J_T(-f_i;\bl_{T-1})$ are the same random variable, and therefore
\begin{multline*}
\Var\{\left[C_T(\bl_{T-1})\right]^{-1/2}J_T(f_i;\bl_{T-1})\}\\ = \cum\left\{\left[C_T(\bl_{T-1})\right]^{-1/2}J_T(f_i;\bl_{T-1}),\left[C_T(\bl_{T-1})\right]^{-1/2}J_T(f_i;\bl_{T-1})\right\} \\
= \cum\left\{\left[C_T(\bl_{T-1})\right]^{-1/2}J_T(f_i;\bl_{T-1}),\left[C_T(\bl_{T-1})\right]^{-1/2}J_T(-f_i;\bl_{T-1})\right\}  
\longrightarrow S(f_i),
\end{multline*}  
as $T\longrightarrow \infty$.
    \end{proof}

\section{Dynamics in gradient descent estimation for AR(2) processes}

In Section 5.2 we show the performance of the forgetting factor Whittle estimator when applied to an AR(2) process with a change-point. In particular, the change-point is present for only $\phi_1$, while $\phi_2$ is kept constant before and after the change-point. However, the figures in the section reveal that the estimates of $\phi_2$ experience a jump as the change-point in $\phi_1$ occurs. More insight is provided by studying the shape the Whittle likelihood function as more data is observed after the change-point. Figure \ref{fig:AR2statespace} shows a realisation of the Whittle likelihood function for the AR(2) process in Section 5.2 as it transitions across the change-point. Recall that the change-point in the process occurs at $T=10000$. In this example we set $\lambda=0.999$. Right at the change-point, when no new data has been observed, the likelihood in maximized at about $(-0.81, -1.46)$. Similarly, after observing enough new data, $300$ data points in this example, the likelihood is maximized at about $(-0.81, 1.46)$. Unsurprisingly, those points are close to the true parameters of the process before and after the change-point, respectively. The interesting behavior occurs when the data in the estimator's `memory' is a mixture of the two processes. The Whittle likelihood of a mixture of two AR(2) processes is maximized by a high value of $\phi_2$ and an in-between value for $\phi_1$. Intuitively, this is due to the shape of the spectral density function of an AR(2) process. It consists of a peak at a central frequency, dependent on $\phi_1$, with the steepness of the peak dependent on $\phi_2$. A low and flat peak fits the spectral density function of the mixture better than a steep peak somewhere in the middle would. This is indeed the path taken by the estimators for $\phi_1$ and $\phi_2$, as seen in dashed line in Figure \ref{fig:AR2statespace} and the corresponding spectral density functions in Figure \ref{fig:S_AR2 comp}.
\begin{figure}[H]
    \centering
    \includegraphics[width=0.8\linewidth]{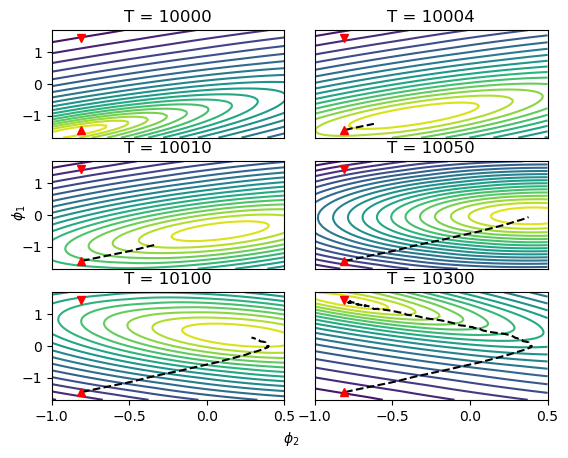}
    %\captionsetup{name=Supplementary Figure}
    \caption{Contour plots for the Whittle likelihood of an AR(2) process as it undergoes a change-point. The change-point occurs at $T=10000$. The dashed line shows estimates for $\phi_1$ and $\phi_2$, computed via the forgetting factor periodogram $\widehat{S}_T(f;\lambda)$, for $T=10000,\dots,10300$. The red triangles denote the true parameter values before and after the change-point.}
    \label{fig:AR2statespace}
\end{figure}
\begin{figure}[H]
    \centering
    \includegraphics[width=0.7\linewidth]{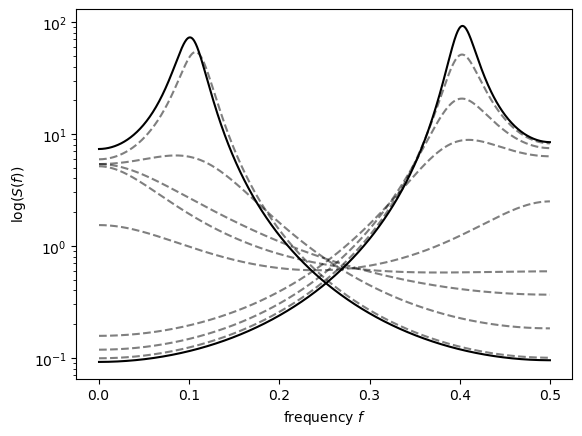}
    %\captionsetup{name=Supplementary Figure}
    \caption{AR(2) spectral density functions derived from the Whittle estimators $\widehat{\phi}_T$, for $T$'s in a representative subset of $\{10000,\dots, 10300\}$. The solid lines show the function for $T=10000$ and $T=10300$. The red triangles denote the true parameter values before and after the change-point.}
    \label{fig:S_AR2 comp}
\end{figure}

\section{Effect of Beta prior on the Adaptive Forgetting Factor Whittle Estimation}

In Section 5.3 we introduce a regularization factor in the Whittle likelihood in the form of a Beta prior on the forgetting factors $\boldsymbol{\lambda}_T$. The intent behind this modification is to induce a predisposition in the algorithm towards a forgetting factor closer to 1. In Figure \ref{fig:alpha effect} the estimated $\boldsymbol{\lambda}_T$'s for three distinct processes are shown. The first two rows are based on AR(2) processes and the bottom row is based on a AR(3) process. The columns indicate whether the Beta prior was used in the estimation procedure, with the left column not including said regularization factor. Note that the top right figure was already presented in the main body on this manuscript as Figure 5. Indeed, the top row of Figure \ref{fig:alpha effect} demonstrates the effect of the Beta prior on the same AR(2) process used for Figures 3-5. In the middle row we observe $\boldsymbol{\lambda}_T$ for an AR(2) process where both $\phi_1$ and $\phi_2$ undergo a change-point, as opposed to only $\phi_1$ in the previous row. Finally, in the bottom row we have the forgetting factor sequence for the same AR(3) process used for Figures 6-7. In all cases, the use of the Beta prior results in a similar dynamic in the sequences; a dip at the change-point quickly followed by the sequence rising back up to mean values between 0.99 and 1.
\begin{figure}[H] 
    \centering
    \includegraphics[width=\linewidth]{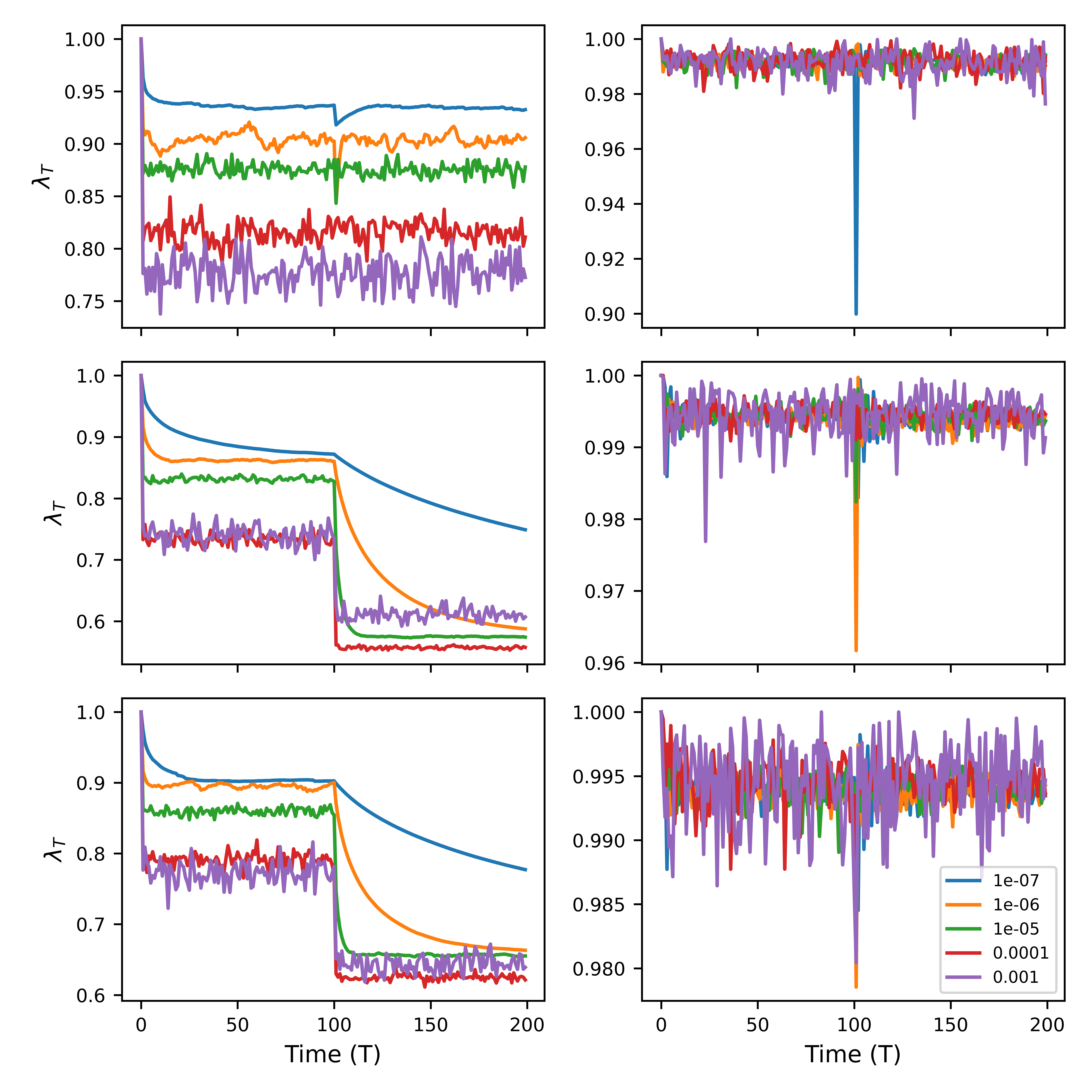}
    %\captionsetup{name=Supplementary Figure}
    \caption{Average forgetting factors $\boldsymbol{\lambda}_T$ associated with the Adaptive Forgetting Factor Whittle Estimation for 2 AR(2) processes (two two rows) and an AR(3) process (bottom row). Figures on the left showcase the average forgetting factor sequences when the Beta prior introduced in Section 5.3 is not used. Whereas, the figures on the right present $\boldsymbol{\lambda}_T$ when said regularization is included in the optimization procedure. Averages obtained via 1000 Monte Carlo simulations.}
    \label{fig:alpha effect}
\end{figure}

\end{appendices}

\bibliography{bibliography.bib}

\end{document}